\pdfoutput=1 

\documentclass[11pt]{article}
\usepackage[numbers,sort&compress]{natbib}

\usepackage[a4paper,margin=2.5cm]{geometry}

\usepackage[T1]{fontenc}
\usepackage[utf8]{inputenc}
\usepackage{lmodern}
\usepackage{microtype}

\usepackage{amsmath,amssymb,amsthm}
\usepackage{mathtools}
\usepackage{bm}
\usepackage{braket}

\usepackage{booktabs}
\usepackage{enumitem}
\usepackage{float}

\usepackage{algorithm}
\usepackage{algpseudocode}

\usepackage{xcolor}

\usepackage[colorlinks=true,linkcolor=blue!60!black,citecolor=green!40!black,urlcolor=blue!70!black]{hyperref}
\hypersetup{
  pdftitle={Partial Information Decomposition of Electronic Observables Along a Reaction Coordinate},
  pdfauthor={Kyunghoon Han and Miguel Gallegos Gonzales},
  pdfcreator={LaTeX},
  pdfproducer={pdflatex}
}

\newcommand{\qX}{\bar{q}_X}
\newcommand{\qY}{\bar{q}_Y}
\newcommand{\R}{\mathcal{R}}
\newcommand{\Imin}{I_{\min}}
\newcommand{\Ispec}{I_{\mathrm{spec}}}

\newcommand{\piT}{\pi_T}
\newcommand{\piX}{\pi_X}

\newtheorem{proposition}{Proposition}
\newtheorem{corollary}{Corollary}
\newtheorem{remark}{Remark}

\newcommand{\tbd}{\begin{center}\itshape [To be written.]\end{center}}

\begin{document}

\title{Partial Information Decomposition of Electronic Observables Along a Reaction Coordinate}

\author{%
  Kyunghoon Han\thanks{Department of Physics and Materials Science, University of Luxembourg, Grand Duchy of Luxembourg. Corresponding author: Kyunghoon Han (email: \texttt{kyunghoon.h@gmail.com})} \and
  Miguel Gallegos\footnotemark[1]
}

\date{}

\maketitle

\begin{abstract}
A reaction-coordinate--resolved information-theoretic analysis of chemical reactivity is developed using mutual information and partial information decomposition (PID).
Along an intrinsic reaction coordinate (IRC), a local empirical distribution is constructed at each position $s$ that couples a coarse-grained geometric progress variable (target) to two electronic readouts (sources), and the joint mutual information $I(T;X,Y)$ is decomposed into redundant, unique, and synergistic contributions using the Williams--Beer PID formalism.
In the numerical demonstrations, the target is a binned bond-asymmetry coordinate $\xi=d_{\mathrm{C}\!-\!\mathrm{nuc}}-d_{\mathrm{C}\!-\!\mathrm{LG}}$, while the sources are DDEC6 net atomic charges on the nucleophile and leaving-group centres.
Application to three prototypical S\textsubscript{N}2 reactions (identity exchange $\mathrm{F^-+CH_3F}$, halide substitution $\mathrm{F^-+CH_3Br}$, and hydroxide substitution $\mathrm{OH^-+CH_3CH_2Br}$) yields compact, symmetry-sensitive signatures of bonding evolution: the identity reaction exhibits mirror-related information profiles with exchange of unique-information contributions between equivalent centres, whereas asymmetric reactions show shifted, centre-specific redistribution among redundancy and synergy as C--X cleavage couples to C--Nu formation.
This Supplementary Information provides the formal construction, chemically motivated limiting toy models, a solvable analytic symmetric-transfer model, and the computational protocol used to obtain IRC-resolved PID curves.
\end{abstract}

\tableofcontents
\bigskip

\IfFileExists{intro.tex}{
\section{Introduction}
\label{sec:intro}

This Supplementary Information expands on material presented at the 2026 Deutsche Physikalische Gesellschaft (DPG, German Physical Society) meeting held in Dresden, Germany (Session CPP 29: Emerging Topics in Chemical and Polymer Physics, New Instruments and Methods III) and provides the detailed theoretical and computational context for a reaction-coordinate--resolved partial information decomposition (PID) analysis of chemical reactivity.

Understanding how bonds form and break is a core objective of theoretical and computational chemistry.
Within the Born--Oppenheimer picture, elementary reactions are commonly analysed on an electronic potential-energy surface in terms of reactant and product basins connected by a saddle point, with mechanistic insight extracted from minimum-energy paths (MEPs) and the intrinsic reaction coordinate (IRC) formalism.
Along such paths, a wide range of established quantum-chemical tools are routinely used to rationalise bonding evolution: population analyses and atomic charges, bond indices and bond orders, valence-bond and resonance pictures, energy-decomposition analyses, and real-space/topological analyses of the electron density.
These ``vanilla'' descriptors provide chemically interpretable narratives of nucleophilic attack, transition-state structure, charge rearrangement, and leaving-group departure, and they remain the standard language for discussing reaction mechanisms across electronic-structure methods.

A recurring practical challenge, however, is that mechanistic interpretation typically relies on \emph{multiple} descriptors, often attached to different atomic centres or defined by different partitions of the same electron density.
Even when these descriptors vary smoothly along an IRC, it is not obvious whether they report essentially the \emph{same} mechanistic signal (redundant readouts), whether one descriptor contains information that the others do not (unique readouts), or whether the mechanistically relevant signal only emerges as a \emph{joint pattern} across descriptors (synergistic readouts).
This ambiguity becomes especially pronounced in reactions where several coupled motions and electronic rearrangements occur simultaneously, and where symmetry (or its breaking) can redistribute electronic response across centres.

Information theory provides a complementary way to quantify statistical dependence between variables without assuming linear response or committing to a particular functional relationship.
In this spirit, mutual information offers a natural quantitative answer to questions of the form: \emph{how strongly does an electronic descriptor correlate with a chosen notion of reaction progress?}
However, mutual information alone cannot distinguish redundancy from synergy when multiple descriptors are considered jointly: two sources can each correlate strongly with a target while carrying largely overlapping information, or they can be individually weak yet jointly highly informative.

The PID framework was introduced precisely to resolve this ambiguity by decomposing the joint mutual information $I(T;X,Y)$ between a \emph{target} $T$ and two \emph{sources} $X,Y$ into four nonnegative components: redundancy, two unique informations, and synergy.
In this work we adopt the original lattice-based definition of PID proposed by Williams and Beer~\cite{Williams2010}, which provides a constructive, nonnegative decomposition and has become a widely used reference point in the PID literature.
At the same time, PID is an active research area with multiple formalisms and axiomatic variants; we cite Kolchinsky~\cite{Kolchinsky2022} as a recent overview that places different approaches in a common conceptual frame.
Our use of PID is therefore deliberately pragmatic: we employ the Williams--Beer decomposition as a consistent bookkeeping of how two electronic readouts encode a chosen geometric macrostate along a reaction coordinate.

Our central proposal is to apply PID \emph{locally along a reaction coordinate}.
At each position $s_0$ on an IRC, we construct an empirical distribution coupling a coarse-grained notion of progress (the target $T$) to two electronic readouts (sources $X,Y$), and we evaluate the mutual information and PID components as functions of $s_0$.
This yields reaction-coordinate--resolved profiles of redundancy, unique information, and synergy, which quantify \emph{how} electronic structure encodes reaction progress and \emph{where} along the path the encoding changes character.
For symmetric reactions, these profiles are constrained by permutation symmetries; for asymmetric reactions, the breakdown of these constraints provides a principled way to localise centre-specific electronic rearrangement.

The remainder of this Supplementary Information is organised as follows.
Section~\ref{sec:theory} introduces the reaction-coordinate--resolved PID construction, including local sampling and discretisation, and fixes the Williams--Beer PID conventions used throughout.
Section~\ref{sec:pid-irc} provides short ``toy'' limiting cases to anchor chemical intuition for redundancy, uniqueness, and synergy.
Section~\ref{sec:toy-model} develops a solvable analytic model for a symmetric atom-transfer motif, illustrating how symmetry and electronic mixing shape PID profiles in a controlled setting.
Finally, the Section \ref{sec:results} benchmarks the framework on prototypical S\textsubscript{N}2 reactions, using density-derived atomic populations (DDEC6 charges) as electronic sources and a bond-asymmetry coordinate as the target.}{\tbd}
\IfFileExists{theory-3.tex}{
\section{Theory}
\label{sec:theory}

\subsection{Partial information decomposition}
\label{sec:pid-defs}

Let $T$, $X$, $Y$ be discrete random variables with joint distribution $p(t,x,y)$.
The \emph{mutual information} between $T$ and a source $X$ is
\begin{equation}
  I(T;X) = \sum_{t,x} p(t,x)\,\log_2\frac{p(t,x)}{p(t)\,p(x)},
  \label{eq:mi}
\end{equation}
and the \emph{joint mutual information} between $T$ and the pair $(X,Y)$ is
\begin{equation}
  I(T;X,Y) = \sum_{t,x,y} p(t,x,y)\,\log_2
  \frac{p(t,x,y)}{p(t)\,p(x,y)}.
  \label{eq:mi-joint}
\end{equation}
Adding a second source cannot decrease mutual information (since conditioning reduces entropy), so $I(T;X,Y) \ge \max\!\big(I(T;X),\; I(T;Y)\big)$.

The partial information decomposition (PID) of Williams and Beer~\cite{Williams2010} decomposes $I(T;X,Y)$ into four terms:\footnote{For the two-source Williams--Beer $\Imin$ decomposition, $\R \ge 0$ and $U_X, U_Y \ge 0$ hold by construction; $S \ge 0$ was shown by Williams and Beer~\cite{Williams2010} for this specific definition (see also~\cite{Bertschinger2014}), though non-negativity is \emph{not} guaranteed for all PID definitions.
Finite-sample estimation, kernel smoothing, and discretisation can produce small numerical negatives.
The implementation records unclipped values (preserving the identity $I(T;X,Y) = \R + U_X + U_Y + S$ exactly).
For display only, any component in a user-specified tolerance window $[-\varepsilon,0]$ is set to zero; the plotted synergy is then defined by the residual $S_{\mathrm{plot}} \equiv I(T;X,Y) - \R_{\mathrm{plot}} - U_{X,\mathrm{plot}} - U_{Y,\mathrm{plot}}$, so that the four displayed components sum to~$I(T;X,Y)$.
Values below $-\varepsilon$ are flagged as unreliable rather than silently clipped.}
\begin{equation}
  I(T;X,Y) = \R + U_X + U_Y + S,
  \label{eq:pid}
\end{equation}
where $\R$ is the \emph{redundancy} (information available from either source alone), $U_X$ and $U_Y$ are the \emph{unique informations} (available from one source but not the other), and $S$ is the \emph{synergy} (available only from the pair).
These four terms are defined via the \emph{specific information}
\begin{equation}
  \Ispec(t;X)
  = \sum_x p(x \mid t)\,\log_2 \frac{p(t \mid x)}{p(t)}
  = \sum_x p(x \mid t)\,\log_2 \frac{p(x \mid t)}{p(x)}
  = D_{\mathrm{KL}}\!\big(p(X \mid t)\;\|\;p(X)\big),
  \label{eq:Ispec}
\end{equation}
which measures the information that observing $X$ provides about a particular outcome $T = t$.
(The second equality uses Bayes' rule, $p(t \mid x)/p(t) = p(x \mid t)/p(x)$.)
Since $\Ispec(t;X) \ge 0$ for every~$t$ (by the non-negativity of KL divergence) and $I(T;X) = \sum_t p(t)\,\Ispec(t;X)$, the Williams--Beer redundancy is
\begin{equation}
  \R \;\equiv\; \Imin(T; X, Y)
  \;=\; \sum_t p(t)\,\min\!\big(\Ispec(t;X),\;\Ispec(t;Y)\big).
  \label{eq:Imin}
\end{equation}
The remaining components follow algebraically:
\begin{equation}
  U_X = I(T;X) - \R,\qquad
  U_Y = I(T;Y) - \R,\qquad
  S   = I(T;X,Y) - I(T;X) - I(T;Y) + \R.
  \label{eq:UUXYS}
\end{equation}
Because the minimum in~\eqref{eq:Imin} never exceeds either specific information,
$U_X \ge 0$ and $U_Y \ge 0$.
It also follows that $\R \le \min\!\big(I(T;X),\, I(T;Y)\big)$, so the redundancy cannot exceed the weaker source's total mutual information.
Throughout this work we use the Williams--Beer $\Imin$ redundancy in the numerical pipeline described in Section~\ref{sec:computational}.
Alternative PID definitions (e.g.\ the Blackwell-order approach of Kolchinsky~\cite{Kolchinsky2022} or other unique-information frameworks~\cite{Bertschinger2014}) could be substituted into the same pipeline by replacing the redundancy functional~$\R$; a comparative study of redundancy measures for chemical applications is an interesting direction for future work.

\subsection{PID along a reaction coordinate}
\label{sec:pid-irc}
Chemical reactions couple nuclear rearrangements to a reorganisation of the electronic structure.
The goal of this work is to quantify, along a reaction path, \emph{which} electronic descriptors carry information about the evolving molecular geometry and \emph{when} two descriptors must be combined to explain that geometry (synergy) versus when they convey overlapping information (redundancy).

Consider a chemical reaction whose progress is assumed to proceed as dictated by a minimum energy path (MEP), conveniently parameterised by the intrinsic reaction coordinate (IRC), here denoted as $s$ (commonly denoted as $\chi$ in chemistry literature, and denoted as $\xi$ to avoid the ambiguity in Section \ref{sec:toy-model}).
At each IRC frame~$i$ with coordinate~$s_i$, a quantum-chemical calculation produces a set of electronic observables (atomic charges, bond orders, density-derived indices, etc.)\ alongside the molecular geometry.
From these data a probability space over reaction ``samples'' (IRC frames, or an ensemble of structures near each $s$) can be defined.
Let $\Omega=\{1,\dots,M\}$ index the sampled structures and let $p(i)$ be a chosen weighting (e.g.\ uniform, or reweighted).
Each discrete variable is then a measurable coarse-graining of a deterministic function of the geometry:
\begin{itemize}
  \item[(i)] \textbf{Geometric target.}
  Choose a feature map $g:\mathbb{R}^{3N}\to\mathbb{R}^{d}$ on nuclear configurations $\mathbf{R}$ (with $d\ge 1$),
  and a discretisation (binning/clustering) map $b_T:\mathbb{R}^d\to\mathcal{T}$ to a finite set of macrostates.
  The geometric target is the ndefined as:
  \begin{equation}
    T(i)\;:=\; b_T\!\big(g(\mathbf{R}_i)\big)\in\mathcal{T}.
  \end{equation}
  This formulation does \emph{not} assume that reaction progress is captured by a single ``key distance''.
  If several coupled motions are essential, one takes $d>1$ and lets $g(\mathbf{R})$ collect the relevant internal coordinates
  (e.g. a vector of distances/angles/dihedrals in the reactive region), or simply $g(\mathbf{R})=s$ (IRC coordinate).
  Mechanisms with changing bonding patterns are handled by defining $g$ to include \emph{all} candidate geometric features
  for the reactive event (rather than selecting one a priori) and then letting $b_T$ define macrostates in this $d$-dimensional
  feature space (e.g. via quantile binning or clustering).

  \item[(ii)] \textbf{Electronic sources.}
    Choose electronic readout functions $e_X,e_Y:\mathbb{R}^{3N}\to\mathbb{R}$ (or $\mathbb{R}^m$) evaluated at each frame (e.g.\ a real-valued atom-centred population-analysis \emph{descriptor} such as a net atomic charge, a bond order, or another density-derived index), together with discretisations $b_X:\mathbb{R}^m\to\mathcal{X}$ and $b_Y:\mathbb{R}^m\to\mathcal{Y}$.
    We then define the discrete sources by coarse-graining these numerical outputs:
    \begin{equation}
      X(i)\;:=\; b_X\!\big(e_X(\mathbf{R}_i)\big)\in\mathcal{X},\qquad
      Y(i)\;:=\; b_Y\!\big(e_Y(\mathbf{R}_i)\big)\in\mathcal{Y}.
    \end{equation}
    In particular, when we write ``atomic charge'' we mean the \emph{numerical net charge returned by a specified population-analysis scheme} (in our computations, the DDEC6 net atomic charge), which is a model-dependent summary of the electron density and is therefore treated here as a descriptor to be coarse-grained rather than as a unique physical observable.
\end{itemize}

The PID of $(T,X,Y)$ quantifies how the chosen \emph{descriptor channels} encode the geometric macrostate $T$:
\[
  \mathbf{R}\;\mapsto\; e(\mathbf{R}) \;\mapsto\; b\!\big(e(\mathbf{R})\big).
\]
Different descriptor choices (and, for charges, different population-analysis schemes) define different channels and can change the numerical PID components.
Accordingly, the decomposition should be interpreted as characterising the information about $T$ that is \emph{accessible through the selected descriptor family}
under the stated partitioning convention.
In practice we assess robustness by (a) comparing chemically motivated alternatives (e.g.\ different charge schemes or charges versus bond indices/density-based measures)
and (b) using discretisations that reduce sensitivity to smooth reparameterisations (e.g.\ quantile binning for monotone transforms), while keeping $T$ fixed.

Because the IRC is one-dimensional, any geometric observable is a deterministic function of $s$.  However, this determinism only forces the PID to collapse if the discretised target $T$ is sufficiently fine that it is (effectively) injective at the sampled resolution.  The following proposition makes this precise.

\begin{proposition}[Deterministic-source redundancy]
\label{prop:det-source}
Let $T$, $X$, $Y$ be discrete random variables.
If $T$ is a deterministic function of~$X$ (i.e.\ $H(T \mid X) = 0$) and also a deterministic function of~$Y$ ($H(T \mid Y) = 0$), then
\begin{equation}
  \R = I(T;X) = I(T;Y) = H(T),\qquad
  U_X = U_Y = S = 0.
\end{equation}
A sufficient condition is that $X = f(T)$ and $Y = g(T)$ for injective~$f,g$, since injectivity makes each map invertible on the support.
\end{proposition}

\begin{proof}
$H(T \mid X) = 0$ implies $I(T;X) = H(T)$.
Moreover, for every $t$, if $p(x\mid t)>0$ then $p(t\mid x)=1$, hence
$\Ispec(t;X)=\log_2(1/p(t))$.
Hence $\min(\Ispec(t;X), \Ispec(t;Y)) = \log_2(1/p(t))$ for all~$t$, giving $\R = H(T)$.
Since $T$ is determined by $X$ alone it is \emph{a fortiori} determined by the pair $(X,Y)$, so $H(T \mid X,Y) = 0$ and $I(T;X,Y) = H(T)$.
Then $U_X = I(T;X) - \R = 0$, $U_Y = 0$, and $S = I(T;X,Y) - I(T;X) - I(T;Y) + \R = H(T) - H(T) - H(T) + H(T) = 0$.
\end{proof}

\begin{remark}[Determinism and coarse-graining]
\label{rem:geom-sources}
Along an IRC, all computed observables---geometric and electronic alike---are, to numerical accuracy, deterministic functions of the nuclear geometry for a fixed electronic state and level of theory.
Non-trivial PID terms arise here not from intrinsic stochasticity but from two sources of coarse-graining.
First, the target~$T$ is a coarse summary of the full microstate: many IRC frames may share a single target bin, so the discretised sources are not in general functions of the discretised target.
Second, the kernel-smoothed local distribution~\eqref{eq:local-joint} pools several neighbouring frames, introducing distributional variation even when every individual frame is deterministic.
Proposition~\ref{prop:det-source} therefore applies only in the limiting case where $T$ is fine enough that $H(T \mid X) = H(T \mid Y) = 0$ holds within each kernel window---i.e.\ each source bin maps to a unique target bin.
In practice, a physically informative decomposition is obtained when the sources vary non-trivially across frames that share the same target bin---a condition most easily satisfied when $X$ and $Y$ are electronic observables (e.g.\ atomic charges) whose dependence on the geometric target is non-monotone or many-to-one at the chosen bin resolution.
\end{remark}

To study how the PID varies along the reaction path, a \emph{local} joint distribution is constructed at each query point~$s_0$.
Concretely, one defines a random experiment: draw a frame index~$i$ with probability proportional to $w_i(s_0)$ and output the triple $(T_i, X_i, Y_i)$.
Given $N$~IRC frames $\{(s_i, T_i, X_i, Y_i)\}_{i=1}^N$ with frame weights $\{\kappa_i\}$, a Gaussian kernel of bandwidth~$h$ defines the local weight $w_i(s_0) = \kappa_i\,\exp\!\big(-(s_0 - s_i)^2 / 2h^2\big)$, and the local joint distribution is
\begin{equation}
  p(t,x,y \mid s_0) = \frac{%
    \sum_i w_i(s_0)\;\mathbf{1}[T_i = t,\, X_i = x,\, Y_i = y]}{%
    \sum_i w_i(s_0)}.
  \label{eq:local-joint}
\end{equation}
The denominator ensures $\sum_{t,x,y} p(t,x,y \mid s_0) = 1$, so~\eqref{eq:local-joint} defines a valid probability distribution for every~$s_0$.
The local PID at~$s_0$ is then obtained by applying equations~\eqref{eq:mi}--\eqref{eq:UUXYS} to this distribution.
Because the distribution itself changes with~$s_0$ (via the kernel window), the resulting partial information (PI) parts are genuine functions of the query point, not smoothed versions of a single global decomposition.
When the two IRC branches have been given branch-equalising weights $\kappa_i = 1/N_b$ (see Step~1 below), the kernel bandwidth~$h$ is the only remaining free parameter controlling the spatial resolution of the decomposition.

\begin{proposition}[Reaction symmetry]
\label{prop:symmetry}
If the reaction possesses a permutation symmetry
$s \mapsto -s$, $T \mapsto \piT(T)$, $X \leftrightarrow Y$
(where $\piT$ is the induced target-bin relabelling), and if this symmetry is respected by the weighting (i.e.\ $p(t,x,y \mid s_0) = p(\piT(t),y,x \mid {-s_0})$), then the local PID at~$s_0$ and the local PID at~$-s_0$ are related by source exchange:
\begin{equation}
  I(T;X)\big|_{s_0} = I(T;Y)\big|_{-s_0},\qquad
  U_X\big|_{s_0} = U_Y\big|_{-s_0},\qquad
  \R\big|_{s_0} = \R\big|_{-s_0},\qquad
  S\big|_{s_0} = S\big|_{-s_0}.
  \label{eq:symmetry-local}
\end{equation}
In particular, at the symmetric point $s_0 = 0$: $I(T;X) = I(T;Y)$ and $U_X = U_Y$.
\end{proposition}

\begin{proof}
By hypothesis, the local distribution satisfies
$p(t,x,y \mid s_0) = p(\piT(t),\,y,\,x \mid {-s_0})$.
Marginalising over~$y$ gives
$p_{TX}(t,x \mid s_0) = p_{TY}(\piT(t),x \mid {-s_0})$.
Since mutual information is invariant under bijective relabelling of states, $I(T;X)|_{s_0} = I(T;Y)|_{-s_0}$.
For the redundancy, the symmetry hypothesis implies
$p(t \mid s_0) = p(\piT(t) \mid {-s_0})$ and
$\Ispec(t;X)|_{s_0} = \Ispec(\piT(t);Y)|_{-s_0}$ (and likewise with $X$ and $Y$ exchanged).
Therefore
\begin{align*}
  \R\big|_{s_0}
  &= \textstyle\sum_t p(t \mid s_0)\,
     \min\!\bigl(\Ispec(t;X)|_{s_0},\;\Ispec(t;Y)|_{s_0}\bigr)\\
  &= \textstyle\sum_t p(\piT(t) \mid {-s_0})\,
     \min\!\bigl(\Ispec(\piT(t);Y)|_{-s_0},\;\Ispec(\piT(t);X)|_{-s_0}\bigr).
\end{align*}
Since $\piT$ is a bijection (and $\min$ is symmetric in its arguments), relabelling $t' = \piT(t)$ gives $\R|_{s_0} = \R|_{-s_0}$.
Then $U_X|_{s_0} = I(T;X)|_{s_0} - \R|_{s_0} = I(T;Y)|_{-s_0} - \R|_{-s_0} = U_Y|_{-s_0}$.
For the joint mutual information, the symmetry hypothesis gives a bijective relabelling of the full joint table $p(t,x,y \mid s_0) \mapsto p(\piT(t),y,x \mid {-s_0})$; since $I(T;X,Y)$ is invariant under bijective relabelling of states, $I(T;X,Y)|_{s_0} = I(T;X,Y)|_{-s_0}$.
The synergy identity then follows from $S = I(T;X,Y) - I(T;X) - I(T;Y) + \R$ and the established identities for each term.
At $s_0 = 0$ the two sides coincide, giving $I(T;X) = I(T;Y)$ and $U_X = U_Y$.
\end{proof}

\noindent
This proposition applies, for instance, to the symmetric hydrogen exchange reaction $\mathrm{H} + \mathrm{H}_2 \to \mathrm{H}_2 + \mathrm{H}$, where the two terminal hydrogen atoms are related by the reaction symmetry.
Note that the proposition establishes $U_X = U_Y$ but does not imply that the unique informations vanish; generically, $U_X = U_Y > 0$ when the target~$T$ is a \emph{signed} observable (e.g.\ $T = \mathrm{bin}(d_X - d_Y)$), because each source is preferentially informative about the side of the reaction on which its atom forms a bond.

\paragraph{General structure of the local PID.}
For an arbitrary reaction---including asymmetric processes such as $\mathrm{A} + \mathrm{BC} \to \mathrm{AB} + \mathrm{C}$ where no permutation symmetry relates the sources---the local PID at every query point~$s_0$ satisfies two identities that follow directly from the definitions~\eqref{eq:pid} and~\eqref{eq:UUXYS} and hold without any symmetry assumption.
First, the difference between the two marginal mutual informations is carried entirely by the unique informations:
\begin{equation}
  I(T;X)\big|_{s_0} - I(T;Y)\big|_{s_0}
  \;=\; U_X(s_0) - U_Y(s_0).
  \label{eq:mi-diff-unique}
\end{equation}
This follows immediately from $I(T;X) = \R + U_X$ and $I(T;Y) = \R + U_Y$; the redundancy cancels exactly.
Equation~\eqref{eq:mi-diff-unique} shows that whenever one source carries more marginal mutual information than the other, the \emph{difference} is accounted for entirely by the unique informations; redundancy and synergy do not contribute to it. (This constrains the difference $U_X - U_Y$; it does not, by itself, determine the absolute magnitudes of $U_X$ and $U_Y$, which also depend on the redundancy.)
Second, the \emph{co-information} $\mathrm{CI}(T;X;Y) \equiv I(T;X) + I(T;Y) - I(T;X,Y)$ decomposes as
\begin{equation}
  \mathrm{CI}(T;X;Y)\big|_{s_0}
  \;=\; \R(s_0) - S(s_0),
  \label{eq:coinfo}
\end{equation}
which follows by substituting $I(T;X) = \R + U_X$, $I(T;Y) = \R + U_Y$, and $I(T;X,Y) = \R + U_X + U_Y + S$ into the definition.
Thus the co-information, which is the standard (but ambiguous) measure of three-variable interaction, is precisely the redundancy--synergy balance:
$\mathrm{CI} > 0$ when redundancy dominates,
$\mathrm{CI} < 0$ when synergy dominates, and
$\mathrm{CI} = 0$ when the two are exactly matched.
The PID resolves this sign ambiguity by reporting $\R$ and $S$ separately.

\paragraph{Chemical interpretation.}
In the molecular context, the four PI parts admit the following interpretation at each point~$s_0$ along the IRC.
These statements concern the statistical structure of the local empirical distribution~\eqref{eq:local-joint} (with a fixed choice of target $T$ and sources $X,Y$); they are suggestive of mechanism but do not, by themselves, establish causality.
Importantly, the numerical values of the components are \emph{conditional} on modelling choices: (i) which electronic descriptors are used to define $X$ and $Y$ (e.g.\ atomic charges, bond indices, density-derived measures), and (ii) how the electronic degrees of freedom are partitioned into ``sources'' (choice of atoms/fragments/orbitals) and, when relevant, which population/partitioning scheme is used.
Accordingly, mechanistic conclusions should be supported by robustness checks across reasonable descriptor families, partition choices, and discretisations.

\begin{enumerate}[label=(\roman*),nosep]
  \item \emph{Redundancy}~$\R(s_0)$ quantifies the geometric information about $T$ that is available from \emph{either} source alone.
  High redundancy indicates that the two electronic readouts respond to the local geometry in a statistically similar manner.
  This can occur, for example, when a conservation law or strong constraint couples the two descriptors (e.g.\ net-charge conservation coupling two atomic charges, or sum rules linking related density partitions), so that each source carries largely overlapping information about the same geometric macrostate. 

  \item \emph{Unique information}~$U_X(s_0)$ [respectively $U_Y(s_0)$] measures what source $X$ [source $Y$] reveals about the target that the other source does not.
  By~\eqref{eq:mi-diff-unique}, the source carrying larger marginal mutual information with the target at~$s_0$ also carries the larger unique information.
  In chemically asymmetric settings, distinct profiles $U_X(s)$ and $U_Y(s)$ are consistent with a picture in which the most informative electronic rearrangement (as represented by the chosen descriptors) is localised more strongly on one partition than the other as the reaction progresses.
  Because ``source'' is defined by a partition of the electronic structure, this interpretation should be checked for stability under alternative but chemically sensible partitions (e.g.\ atom-vs-fragment) and partitioning schemes.

  \item \emph{Synergy}~$S(s_0)$ captures geometric information about $T$ that is inaccessible from either source alone but becomes available when both are observed jointly.
  A peak in $S(s_0)$ is consistent with a regime where reaction progress is encoded primarily in a \emph{joint electronic pattern} (e.g.\ a coupled change across two atoms/fragments, or complementary behaviour of two descriptors) rather than in either source individually.
  While this may reflect cooperative electronic reorganisation, synergy can also arise from non-mechanistic effects such as discretisation artefacts (bin-edge sensitivity), strong constraints that couple the sources, or target definitions that require combining complementary contributions from each source.
  Disentangling these possibilities requires comparison across bin resolutions, alternative target maps $g(\mathbf{R})$, and alternative descriptor/partition choices.
\end{enumerate}

\paragraph{Chemically motivated limiting cases.}
To anchor intuition with a chemically interpretable construction, consider an association motif $A^+ + B^- \rightarrow A$--$B$ and let $n_A(\mathbf{R}),n_B(\mathbf{R})$ denote any partitioned electronic populations on centres $A$ and $B$ (equivalently net charges $q_{A,B}=Z_{A,B}-n_{A,B}$), referenced to the separated-fragment values.
More precisely, fix a reference geometry $\mathbf{R}_{\mathrm{ref}}$ on the reactant-side asymptote (separated fragments) and define the population deviations
$\Delta n_A(\mathbf{R}) := n_A(\mathbf{R})-n_A(\mathbf{R}_{\mathrm{ref}})$ and $\Delta n_B(\mathbf{R}) := n_B(\mathbf{R})-n_B(\mathbf{R}_{\mathrm{ref}})$ (equivalently $\Delta q_A=-\Delta n_A$, $\Delta q_B=-\Delta n_B$).
Define binary electronic readouts by the \emph{sign} of these deviations,
\[
X:=\mathbf{1}\!\left[\Delta n_A>0\right],\qquad
Y:=\mathbf{1}\!\left[\Delta n_B>0\right],
\]
and let $T\in\{0,1\}$ denote a coarse-grained geometric/electronic macrostate (equiprobable for illustration).
Then three limiting mechanisms correspond to explicit conditional laws $p(x,y\mid t)$:
\begin{itemize}
  \item[(i)] \emph{Pure redundancy (transfer-dominated association).}
  In the association motif $A^+ + B^- \rightarrow A$--$B$, a natural idealisation is a \emph{two-centre transfer} picture in which the electronic change relevant to $T$ is essentially a redistribution between the two centres.
  To keep the toy model unbiased, let $T$ encode the direction of net transfer ($T=1$ for $A\rightarrow B$, $T=0$ for $B\rightarrow A$), corresponding to opposite-sign deviations $(\Delta n_A,\Delta n_B)$.
  In this idealised limit the transfer is confined to the two-centre partition, so electron-number conservation implies $\Delta n_A=-\Delta n_B$ and hence $(X,Y)$ is deterministic given $T$ (one deviation is positive and the other negative).
  Consequently $I(T;X)=I(T;Y)=I(T;X,Y)=1$ bit and the PID yields $\R=1$ bit with $U_X=U_Y=S=0$, i.e.\ both centres report the same transfer event in an informationally overlapping way.
  (In an actual ionic association $A^+ + B^-$ one transfer direction will dominate, but the redundancy logic is unchanged: either centre alone can diagnose the transfer-like mode.)

  \item[(ii)] \emph{Pure unique information (localised polarisation).}
  A second chemically meaningful limit is \emph{localised} electronic response during approach: the macrostate $T$ (e.g.\ ``associated'' vs.\ ``separated'', or ``activated'' vs.\ ``non-activated'') is reflected primarily in the population of one centre, while the other centre is effectively a spectator at the chosen resolution.
  A minimal model sets $X=T$ deterministically and takes $Y$ independent noise ($p(y\mid t)=\tfrac12$).
  Then $I(T;X)=1$, $I(T;Y)=0$, $I(T;X,Y)=1$, hence $U_X=1$ bit with $\R=U_Y=S=0$.

  \item[(iii)] \emph{Pure synergy (association encoded by a two-centre pattern).}
  Finally, association can be encoded not by a one-centre response but by a \emph{joint pattern} across centres, such as distinguishing a transfer-like (anti-phase) redistribution from a sharing/polarisation-like (in-phase) response.
  In a full molecular environment, the two-centre deviations need not sum to zero because other degrees of freedom (the bonding region or additional atoms/fragments) can supply/accept density; only $\sum_k \Delta n_k=0$ is enforced globally.
  Let $T$ label these two modes ($T=0$ anti-phase, $T=1$ in-phase) and idealise them by
  \[
    p(x,y\mid T=1)=\tfrac12\,\mathbf{1}_{\{x=y\}},\qquad
    p(x,y\mid T=0)=\tfrac12\,\mathbf{1}_{\{x\neq y\}},
  \]
  i.e.\ $T=1$ iff $\Delta n_A$ and $\Delta n_B$ have the same sign and $T=0$ iff they have opposite signs.
  In this construction each marginal is uninformative ($I(T;X)=I(T;Y)=0$), but the joint observation resolves the mode ($I(T;X,Y)=1$), so the PID gives $S=1$ bit with $\R=U_X=U_Y=0$.
\end{itemize}
These toy models show, in chemically interpretable terms, how redundancy corresponds to overlapping one-centre readouts of the same rearrangement, unique information to a predominantly localised response, and synergy to progress encoded primarily in a genuinely two-centre pattern that only becomes visible when both partitions are observed jointly.

\paragraph{Diagnostic value of the profile shapes.}
The variation of the PI parts along~$s$ provides richer diagnostic information than the mutual information $I(T;X,Y)$ alone.
Two regions of the IRC with the same total mutual information may differ qualitatively in their decomposition: a region dominated by redundancy ($\R \gg S$) implies that both sources track the geometry in parallel, whereas a region dominated by synergy ($S \gg \R$) implies that the geometric signal resides primarily in the \emph{joint} source distribution $p(x,y \mid t)$ rather than in either marginal $p(x \mid t)$ or $p(y \mid t)$.
Likewise, a shift from $U_X > U_Y$ to $U_Y > U_X$ along the IRC pinpoints where the dominant electronic rearrangement passes from one atom to the other.
These distinctions are invisible to any single-source or joint-information analysis and motivate the use of PID as a bonding diagnostic.
The symmetric case (Proposition~\ref{prop:symmetry}) is a special instance of this framework in which the constraint $U_X(s) = U_Y(-s)$ and $\R(s) = \R(-s)$ simplifies the profile shapes but does not alter the interpretive logic.

\subsection{Numerical evaluation of the local PID}
\label{sec:computational}

\begin{algorithm}[t]
\caption{Kernel-smoothed PID along an IRC}
\label{alg:pid-irc}
\begin{algorithmic}[1]
\Require IRC frames
  $\{(s_i,\, T_i^{\mathrm{raw}},\, q_{X,i},\, q_{Y,i})\}_{i=1}^{N}$;
  bin counts $n_T, n_X, n_Y$;
  grid size $G$; initial bandwidth $h_0$;
  $\mathrm{ESS}_{\min}$; $h_{\max}$;
  masking threshold $p_{\min}$;
  minimum occupied target bins $m_T$
\Ensure PID curves
  $\{\R(s_0),\, U_X(s_0),\, U_Y(s_0),\, S(s_0)\}$
  on a uniform grid;
  reliability flag per grid point
\Statex
\State Assign branch-equalising weights
  $\kappa_i \gets 1/N_b$  \Comment{$N_b$ = branch frame count for frame~$i$}
\State Bin observables (using globally fixed bin edges):
  $T_i \gets \mathrm{bin}(T_i^{\mathrm{raw}})$,\;
  $X_i \gets \mathrm{bin}(q_{X,i})$,\;
  $Y_i \gets \mathrm{bin}(q_{Y,i})$
\State Encode each frame as a one-hot vector;
  stack into $\mathbf{O} \in \{0,1\}^{N \times (n_T n_X n_Y)}$
\For{each grid point $s_0$ in
     $\mathrm{linspace}(s_{\min}, s_{\max}, G)$}
  \State $h \gets h_0$;\; $\mathrm{reliable} \gets \mathrm{true}$
  \Repeat \Comment{Adaptive bandwidth (Eq.~\ref{eq:ess})}
    \State $w_i \gets \kappa_i\,
      \exp\!\bigl(-(s_0 - s_i)^2 / 2h^2\bigr)$
    \State $\mathbf{w} \gets (w_i)_{i=1}^N$
    \State $\mathrm{ESS} \gets
      (\textstyle\sum_i w_i)^2 \big/ \sum_i w_i^2$
    \State $\mathbf{j} \gets \mathbf{w}^\top \mathbf{O}$;\;
      $p(t,x,y) \gets \mathbf{j}/\|\mathbf{j}\|_1$
      \Comment{Eq.~\ref{eq:local-joint}}
    \State $p(t) \gets \sum_{x,y} p(t,x,y)$
    \State $\mathcal{T} \gets \{t : p(t) \ge p_{\min}\}$
      \Comment{Mask negligible target bins}
    \If{$\mathrm{ESS} < \mathrm{ESS}_{\min}$
        \textbf{or} $|\mathcal{T}| < m_T$}
      \State $h \gets \min(1.25\,h,\; h_{\max})$
    \EndIf
  \Until{($\mathrm{ESS} \ge \mathrm{ESS}_{\min}$
         \textbf{and} $|\mathcal{T}| \ge m_T$)
         \textbf{or} $h = h_{\max}$}
  \If{$|\mathcal{T}| < m_T$}
    \State $\mathrm{reliable} \gets \mathrm{false}$
  \EndIf
  \State Restrict and renormalise:
    $\tilde{p}(t,x,y) \gets p(t,x,y)/Z$ for $t \in \mathcal{T}$,\;
    $\tilde{p}(t,x,y) \gets 0$ for $t \notin \mathcal{T}$,\;
    where $Z = \sum_{t' \in \mathcal{T},\, x,\, y} p(t',x,y)$
    \Comment{Eq.~\ref{eq:renorm}}
  \State Compute $I(T;X)$, $I(T;Y)$, $I(T;X,Y)$
    from marginals of $\tilde{p}$
    \Comment{Eqs.~\ref{eq:mi}--\ref{eq:mi-joint}}
  \State Compute $\Ispec(t;X)$, $\Ispec(t;Y)$ for each $t \in \mathcal{T}$
    from $\tilde{p}$
    \Comment{Eq.~\ref{eq:Ispec}}
  \State $\R \gets \sum_{t \in \mathcal{T}} \tilde{p}(t)\,
    \min(\Ispec(t;X),\, \Ispec(t;Y))$
    \Comment{Eq.~\ref{eq:Imin}}
  \State $U_X \gets I(T;X) - \R$;\,
         $U_Y \gets I(T;Y) - \R$;\,
         $S \gets I(T;X,Y) - I(T;X) - I(T;Y) + \R$
         \Comment{Eq.~\ref{eq:UUXYS}}
\EndFor
\end{algorithmic}
\end{algorithm}

The local PID of Section~\ref{sec:pid-irc} requires a kernel-smoothed joint distribution at every query point along the IRC, from which the Williams--Beer decomposition is extracted. Algorithm~\ref{alg:pid-irc} summarises the full procedure; the remainder of this subsection explains its key steps and the design choices behind them.
The algorithm is formulated for a general reaction with an arbitrary number of IRC frames; optional specialisations for symmetric reactions are noted where relevant.

\paragraph{Step~1: Arc-length parameterisation (line~1).}
A meaningful kernel smoother requires an IRC coordinate~$s$ whose spacing reflects physical displacement, not the arbitrary step count of the electronic-structure optimiser. 
Each IRC frame~$i$ is therefore assigned a coordinate~$s_i$ equal to the cumulative Cartesian RMSD displacement from the transition state (TS).
The arc length is computed separately on each branch: denoting the TS as frame~0 and the $k$-th frame along a given branch as frame~$k$, consecutive frames are rigidly aligned (translation and rotation removed via Kabsch superposition~\cite{Kabsch1976}) and each increment is accumulated as
\begin{equation}
  |s_k| = \sum_{j=1}^{k}
  \sqrt{\frac{1}{N_{\mathrm{at}}}
        \sum_{\alpha=1}^{N_{\mathrm{at}}}
        |\mathbf{r}_\alpha^{(j)} - \mathbf{r}_\alpha^{(j-1)}|^2}\,,
  \label{eq:arc-length}
\end{equation}
where the Kabsch alignment is applied to each pair $(j-1,j)$ independently (i.e.\ it does not chain across the full branch, avoiding drift).
The sign of $s_k$ is chosen positive on the forward branch and negative on the backward branch, and the two branches are then concatenated with the TS at $s = 0$.
Because the two branches of an IRC calculation need not contain the same number of frames, nor the same step sizes, each frame receives a branch-equalising weight $\kappa_i = 1/N_b$, where $N_b$ is the number of frames in the branch to which frame~$i$ belongs.
Each branch thus carries total weight~1, so that neither branch dominates the kernel estimator regardless of differences in step count or step size.
Optionally, the coordinate mapping from frame index to~$s$ may be resampled onto a uniform $s$-grid by cubic interpolation, yielding a strictly density-equalised dataset.
This resampling affects only the assignment of frames to $s$-coordinates; it does not interpolate or invent new electronic observables, which would introduce unphysical correlations.

\paragraph{Step~2: Discretisation (lines~2--3).}
The continuous observables are mapped to discrete random variables by binning.
The geometric target~$T$ is obtained by equal-width binning of a structural observable that characterises the reaction progress---typically a bond-length difference $\Delta d = d_X - d_Y$ for a bond-transfer reaction, though any suitable geometric descriptor (e.g.\ a dihedral angle, a coordination number) may serve.
The source variables $X = \mathrm{bin}(q_X)$ and $Y = \mathrm{bin}(q_Y)$ are obtained by equal-occupancy (quantile) binning of electronic observables, most commonly net atomic charges from a real-space partitioning scheme such as DDEC6~\cite{Manz2016}, Bader, or Hirshfeld.
Specifically, the bin edges for source~$X$ are the $(k/n_X)$-quantiles of $\{q_{X,i}\}_{i=1}^N$ (weighted by $\{\kappa_i\}$) for $k = 0, 1, \ldots, n_X$, so that each bin contains approximately equal total weight; source~$Y$ is treated analogously.
Quantile bins equalise the marginal occupancy of each source level, reducing sparse-cell bias in the specific-information estimates~\eqref{eq:Ispec}; equal-width bins are used for the target so that the bin centres have a direct geometric interpretation and, in symmetric reactions, the bin edges can be placed symmetrically about zero (Remark~\ref{rem:sym-bins}).
As noted in Remark~\ref{rem:geom-sources}, although all observables along a fixed IRC are deterministic, atomic charges typically have a non-monotone or many-to-one relationship with the geometric target at the chosen bin resolution, so that the discretised sources vary across frames sharing the same target bin---exactly the condition under which the kernel-windowed PID produces non-trivial structure.

A typical resolution is $n_T = n_X = n_Y = 5$, yielding a joint table of $5^3 = 125$ cells---coarse enough to keep statistical noise low at modest effective sample sizes, yet fine enough to resolve major bonding transitions.
Crucially, the bin edges are fixed globally (computed once from the full IRC dataset) and are not adapted per query point~$s_0$; this ensures that the PID values at different~$s_0$ are computed on the same discrete state space and are therefore directly comparable along the reaction path.
Finer grids (e.g.\ $n = 7$ or $9$) are possible when more IRC frames are available but provide diminishing returns once the relevant features of the charge profiles are captured.
Each frame~$i$ is then encoded as a one-hot vector over the $n_T \!\cdot\! n_X \!\cdot\! n_Y$ joint cells (Algorithm~\ref{alg:pid-irc}, line~3); stacking these vectors into a matrix~$\mathbf{O} \in \{0,1\}^{N \times (n_T n_X n_Y)}$ enables the local joint table to be assembled as a single matrix--vector product $\mathbf{j}(s_0) = \mathbf{w}^\top\!\mathbf{O}$ (line~9).

\begin{remark}[Bin construction for symmetric reactions]
\label{rem:sym-bins}
When the reaction possesses the permutation symmetry of Proposition~\ref{prop:symmetry}, two additional bin-construction conventions are advantageous.
First, target bin edges are chosen symmetric about zero, $b_k^T = -\,b_{n_T-k}^T$, which defines the target-bin mirror map~$\piT$.
Second, source bin edges are shared, $b_k^X = b_k^Y \equiv b_k^q$, constructed by pooling all $q_X$ and $q_Y$ values; this defines the source-bin mirror map~$\piX$.
Together these maps are the prerequisites for the symmetry results developed in the analytic model of Section~\ref{sec:toy-model} (Propositions~\ref{prop:U-equal} and~\ref{prop:U-zero-even}).
For an asymmetric reaction, the bin edges for $X$ and $Y$ may be chosen independently to match the distinct ranges of each electronic source.
\end{remark}

\paragraph{Step~3: Adaptive kernel bandwidth (lines~5--17).}
At each grid point~$s_0$, the Gaussian kernel $w_i(s_0) = \kappa_i\,\exp\!\bigl(-(s_0 - s_i)^2 / 2h^2\bigr)$ determines the local weight of every frame; the bandwidth~$h$ has units of arc length.
The bandwidth~$h$ controls the bias--variance trade-off: a narrow kernel resolves sharp features but may leave too few effective frames to populate the joint table reliably; a wide kernel smooths genuine structure.
Rather than fixing~$h$ globally, the algorithm adapts it at each grid point to satisfy two criteria simultaneously: (i)~a minimum Kish effective sample size~\cite{Kish1965},
\begin{equation}
  \mathrm{ESS}(s_0) \equiv
  \frac{\bigl(\sum_i w_i\bigr)^2}{\sum_i w_i^2}
  \;\ge\; \mathrm{ESS}_{\min},
  \label{eq:ess}
\end{equation}
and (ii)~a minimum number~$m_T$ of non-negligible target bins ($|\mathcal{T}| \ge m_T$, where $\mathcal{T} = \{t : p(t \mid s_0) \ge p_{\min}\}$).
Starting from an initial bandwidth $h_0 = 0.02\,(s_{\max} - s_{\min})$ (2\% of the IRC span), $h$~is increased by a factor of~1.25 at each iteration until both criteria are met or a ceiling $h_{\max} = 0.10\,(s_{\max} - s_{\min})$ is reached.
Because the bin-coverage check requires the target marginal $p(t \mid s_0)$, the local joint table~\eqref{eq:local-joint} is assembled inside the adaptive loop (lines~7--10) and re-evaluated at each bandwidth increment.
In practice this means the kernel is narrow near the transition state, where IRC frames are usually densely spaced, and widens automatically toward the reactant and product basins where frames are sparser.
The ESS criterion controls the total effective weight but does not guarantee that every joint cell $(t,x,y)$ is adequately populated; in particular, peripheral target bins may remain sparsely occupied even at large ESS.
Target bins whose marginal weight $p(t \mid s_0)$ falls below a threshold $p_{\min} = 10^{-12}$ are \emph{masked}: they are excluded from the subsequent PID computation.
The secondary check requires at least $m_T$ target bins to carry non-negligible weight (a typical default is $m_T = 2$, the minimum needed for meaningful specific-information comparisons; setting $m_T$ closer to~$n_T$ imposes stricter coverage at the cost of possible oversmoothing near the endpoints).
If this condition cannot be met at~$h_{\max}$, the grid point is flagged as unreliable (line~16).

\paragraph{Step~4: Renormalisation and PID evaluation
(lines~18--22).}
With the adaptive loop converged, the local joint distribution is restricted to the retained target bins~$\mathcal{T}$ and renormalised to a proper probability distribution:
\begin{equation}
  \tilde{p}(t,x,y \mid s_0) =
  \frac{p(t,x,y \mid s_0)}%
       {\sum_{t' \in \mathcal{T}} \sum_{x',y'} p(t',x',y' \mid s_0)},
  \qquad t \in \mathcal{T}.
  \label{eq:renorm}
\end{equation}
All subsequent information-theoretic quantities
\begin{equation*}
    I(T;X),\, I(T;Y),\, I(T;X,Y)\,\Ispec(t;X),\, \text{and } \Ispec(t;Y)
\end{equation*}
are computed from the marginals of~$\tilde{p}$, and the Williams--Beer $\Imin$~\eqref{eq:Imin} yields~$\R$.
Because $\tilde{p}$ is a valid distribution (i.e.\ $\sum_{t,x,y} \tilde{p} = 1$), the decomposition identity $I(T;X,Y) = \R + U_X + U_Y + S$ holds exactly over the retained state space.
Note that the effective target alphabet is $\mathcal{T}$, not $\{1,\ldots,n_T\}$; consequently $H(T)$ computed from~$\tilde{p}$ is bounded by $\log_2|\mathcal{T}|$ rather than $\log_2 n_T$, and PID values at grid points with different $|\mathcal{T}|$ are not strictly on the same scale.
This is preferable to adding a pseudocount, which would introduce a systematic bias toward uniformity in~$\Ispec$; the trade-off is that heavy masking (many empty target bins) effectively reduces the dynamic range of~$T$ and can underestimate both redundancy and mutual information.
The remaining components $U_X$, $U_Y$, $S$ follow from~\eqref{eq:UUXYS} (line~22).
Because the entire batch of grid points shares the same one-hot matrix~$\mathbf{O}$, the full PID curve is obtained by a single weight-matrix--one-hot product followed by vectorised PID evaluations, giving a total complexity of $O(G \cdot N + G \cdot n_T n_X n_Y)$ for $G$~grid points.
For typical parameters ($G = 201$, $N \approx 200$, $n_T = n_X = n_Y = 5$), the entire sweep completes in milliseconds to seconds on a modern workstation.

}{\tbd}
\IfFileExists{h2.tex}{
\section{Analytic model: symmetric atom transfer}
\label{sec:toy-model}

To illustrate the general framework of Section~\ref{sec:theory} in a solvable setting, this section develops a minimal analytic model for a symmetric atom-transfer reaction of the form
$\mathrm{X} + \mathrm{A\text{-}Y} \to \mathrm{X\text{-}A}
+ \mathrm{Y}$,
where A is a central (anchor) atom and X, Y are two equivalent terminal atoms.
The simplest of such reaction is the collinear hydrogen exchange reaction
$\mathrm{H} + \mathrm{H}_2 \to \mathrm{H}_2 + \mathrm{H}$.
While the tri-hydrogen system also admits formally charge-separated (ionic) asymptotes such as $\mathrm{H}^+ + \mathrm{H}_2^{-}$ (and the permuted counterpart $\mathrm{H}^- + \mathrm{H}_2^{+}$), these channels involve metastable anionic/ionic character and, in general, require a multi-state diabatic treatment with non-adiabatic couplings rather than a single-surface closed-form model.
Since our purpose here is a transparent analytic demonstration of the PID construction (not a quantitatively complete description of all reactive/charge-transfer channels), we therefore restrict to the lowest, symmetry-adapted neutral exchange picture along the collinear MEP.

The model is constructed to match the anchor-based geometry target $T = \mathrm{bin}(d_X - d_Y)$ used in the computational pipeline; a pure diatomic (two-centre) reaction has no natural pair of electronic sources and its PID collapses toward pure redundancy unless an independent second electronic observable (e.g.\ bond order or ELF) supplements the charge.

\subsection{Valence-bond Hamiltonian}
\label{sec:hamiltonian}

The reaction coordinate is the bond-length asymmetry
$\xi \equiv d_X - d_Y$, where $d_X = |\mathbf{r}_X - \mathbf{r}_A|$
and $d_Y = |\mathbf{r}_Y - \mathbf{r}_A|$ are the distances from each terminal atom to the anchor.
At the transition state, $\xi = 0$ and $d_X = d_Y$ by symmetry; the reactant basin corresponds to $\xi > 0$ ($\mathrm{X}$ far, $\mathrm{A\text{-}Y}$ bonded) and the product basin to $\xi < 0$ ($\mathrm{X\text{-}A}$ bonded, $\mathrm{Y}$ far).

In the minimal valence-bond (VB) picture~\cite{Shaik2007}, the two dominant electronic configurations are
\begin{align}
  \ket{L} &= \ket{\mathrm{X}^{\,\cdot} \;\cdots\;
              \mathrm{A\text{-}Y}} &
  &\text{(left radical, right bond)},\notag\\
  \ket{R} &= \ket{\mathrm{X\text{-}A} \;\cdots\;
              \mathrm{Y}^{\,\cdot}} &
  &\text{(left bond, right radical)}.
  \label{eq:VB-states}
\end{align}

A complete VB description augments the covalent (radical) structures~\eqref{eq:VB-states}
with charge-separated (ionic) configurations.  In the present notation these may be written schematically as
\begin{align}
  \ket{I_L} &= \ket{\mathrm{X}^{-} \;\cdots\; \mathrm{A\text{-}Y}^{+}} &
  &\text{(left anion, right cation)},\notag\\
  \ket{I_R} &= \ket{\mathrm{X\text{-}A}^{+} \;\cdots\; \mathrm{Y}^{-}} &
  &\text{(left cation, right anion)},\notag
\end{align}
(and permutations thereof; for the collinear $\mathrm{H}+\mathrm{H}_2$ prototype these correspond to the familiar
$\mathrm{H}^{-}+\mathrm{H}_2^{+}$ and $\mathrm{H}^{+}+\mathrm{H}_2^{-}$ resonance forms).
Including $\{\ket{I_L},\ket{I_R}\}$ enlarges the VB Hamiltonian but does not alter the symmetry structure of the atom-transfer problem; when the ionic states are energetically well separated along the neutral MEP, they can be adiabatically eliminated (e.g.\ by L\"owdin/Schur downfolding \cite{schur1917potenzreihen,zhang2006schur,lowdin1964studies,feshbach1962unified}), yielding an effective two-state model in the $\{\ket{L},\ket{R}\}$ subspace with renormalised diabatic energies and couplings.
For clarity and analytic tractability, we therefore work in the minimal covalent subspace, noting that ionic corrections can be incorporated straightforwardly by reparameterising the effective $2\times 2$ model.

The diabatic energy of $\ket{L}$ is lower when $\xi > 0$ (the A--Y bond is intact); that of $\ket{R}$ is lower when $\xi < 0$.
Linearising the energy gap about the crossing $\xi = 0$ gives the Landau--Zener Hamiltonian~\cite{Landau1932,Zener1932} in the basis $\{\ket{L}, \ket{R}\}$:
\begin{equation}
  \hat{H}(\xi) =
  \begin{pmatrix}
    -\Delta\xi/2 & V \\[3pt]
    V & +\Delta\xi/2
  \end{pmatrix},
  \label{eq:hamiltonian}
\end{equation}
where $\Delta > 0$ is the slope of the diabatic energy gap and $V > 0$ is the resonance coupling between the two VB structures. Diagonalisation yields eigenvalues $E_\pm = \pm\sqrt{(\Delta\xi/2)^2 + V^2}$ and a ground-state weight on $\ket{R}$:
\begin{equation}
  w(\xi) = \frac{1}{2}\Bigl(1
  - \frac{\xi}{\sqrt{\xi^2 + \lambda^2}}\Bigr),
  \qquad \lambda \equiv \frac{2V}{\Delta},
  \label{eq:weight}
\end{equation}
where $\lambda$ is the characteristic width of the avoided crossing in $\xi$-space.
The ground state is 
$\ket{\Psi} = \sqrt{1-w}\,\ket{L} + \sqrt{w}\,\ket{R}$,
so $w \to 0$ for $\xi \gg \lambda$ (reactant basin, dominated by $\ket{L}$) and $w \to 1$ for $\xi \ll -\lambda$ (product basin, dominated by $\ket{R}$).

To connect the VB weights to atomic charges, each VB state is assigned charge parameters reflecting the bond polarity.
In the free-radical configuration of atom~X (as in $\ket{L}$), the atom carries a reference charge~$q_0$; in the bonded configuration (as in $\ket{R}$), electron density shifts toward the bond region and the charge becomes $q_0 - \delta$.
The charge transfer parameter $\delta > 0$ encodes the magnitude of electron redistribution upon bond formation.
The ground-state mean charges on the two terminal atoms are then
\begin{align}
  \qX(\xi) &= q_0 - \delta\,w(\xi),\notag\\
  \qY(\xi) &= q_0 - \delta\,\bigl(1 - w(\xi)\bigr),
  \label{eq:charge-profiles}
\end{align}
and the anchor charge is
$\bar{q}_A(\xi) = -2q_0 + \delta$ (constant in the two-state model, fixed by overall charge neutrality $\qX + \qY + \bar{q}_A = 0$; one may set $q_0 = \delta/2$ so that all three mean charges vanish at the transition state without loss of generality).
The charge profiles $\qX$ and $\qY$ are both sigmoids of width $\lambda$, related by the reaction symmetry: $\qX(\xi) = \qY(-\xi)$.

\begin{remark}[Charge-sum constraint]
\label{rem:charge-sum}
In this minimal two-state model, $\qX(\xi) + \qY(\xi) = 2q_0 - \delta$ for all~$\xi$: the sum of the terminal charges is constant, with the complementary charge residing on the anchor.
In a richer electronic model (or in a real quantum-chemical calculation), the anchor charge~$\bar{q}_A$ also varies with $\xi$, relaxing this constraint and introducing genuine independence between $\qX$ and~$\qY$.
\end{remark}

\subsection{Joint distribution and the PID}
\label{sec:joint-pid}

Following the framework of Section~\ref{sec:pid-irc}, the three PID variables are: the target $T = \mathrm{bin}(\xi) \in \{1,\ldots,n_T\}$, discretised with symmetric bin edges
$b_k^T = -\,b_{n_T - k}^T$ for $k = 0,\ldots,n_T$;\;
and the sources
$X = \mathrm{bin}(q_X^{\mathrm{obs}}) \in \{1,\ldots,n_X\}$,
$Y = \mathrm{bin}(q_Y^{\mathrm{obs}}) \in \{1,\ldots,n_Y\}$
with shared symmetric charge-bin edges $b_k^X = b_k^Y \equiv b_k^q$ satisfying $b_k^q = -\,b_{n_X - k}^q$ (after centering on the midpoint charge).
The observed charges are
\begin{equation}
  q_X^{\mathrm{obs}} = \qX(\xi) + \eta_X,\qquad
  q_Y^{\mathrm{obs}} = \qY(\xi) + \eta_Y,\qquad
  \eta_{X,Y} \stackrel{\mathrm{iid}}{\sim}
  \mathcal{N}(0,\sigma_q^2).
  \label{eq:obs-model}
\end{equation}
The noise width $\sigma_q$ models the spread inherent in any charge-partitioning scheme (Bader, DDEC, Hirshfeld, etc.)~\cite{Manz2016} as well as electronic degrees of freedom beyond the two-state model---in particular, the variation of the anchor charge $\bar{q}_A$ that the minimal VB picture holds fixed.
Taking the two noise terms independent is the key modelling choice: it represents the regime in which the two terminal charges carry genuinely independent information, as occurs in practice when the anchor charge varies or when the sources are distinct electronic observables (e.g.\ a charge and a bond order).

At a query point~$\xi_0$ along the reaction coordinate, a Gaussian kernel of bandwidth~$h$ defines the local density $\rho(\xi \mid \xi_0) = \mathcal{N}(\xi_0, h^2)$.
(We keep $h$ fixed in the analytic model; in the numerical pipeline $h$ is chosen adaptively to satisfy a minimum effective sample size.)
Conditional independence of $X$ and $Y$ given~$\xi$ (a consequence of the i.i.d.\ noise model) then allows the local joint probability to factorise into a one-dimensional integral:
\begin{equation}
  P(t,i,j \mid \xi_0) = \int_{I_t}
  \rho(\xi \mid \xi_0)\;
  P(X\!=\!i \mid \xi)\;
  P(Y\!=\!j \mid \xi)\;d\xi,
  \label{eq:joint}
\end{equation}
where $I_t = [b_{t-1}^T, b_t^T)$ is the $t$-th target bin and the conditional bin probabilities are
\begin{align}
  P(X\!=\!i \mid \xi)
  &= \Phi\!\Bigl(\frac{b_i^q - \qX(\xi)}{\sigma_q}\Bigr)
  - \Phi\!\Bigl(\frac{b_{i-1}^q - \qX(\xi)}{\sigma_q}\Bigr),
  \label{eq:PX}\\
  P(Y\!=\!j \mid \xi)
  &= \Phi\!\Bigl(\frac{b_j^q - \qY(\xi)}{\sigma_q}\Bigr)
  - \Phi\!\Bigl(\frac{b_{j-1}^q - \qY(\xi)}{\sigma_q}\Bigr),
  \label{eq:PY}
\end{align}
with $\Phi$ the standard normal CDF.
This factorisation is a modelling assumption: in real electronic structure calculations, residual degrees of freedom that cause the charges to fluctuate also tend to correlate them (and charge conservation pushes toward anti-correlation).
Relaxing the conditional independence introduces a nontrivial copula $P(X,Y \mid \xi)$ and changes the PID quantitatively.

The integrand is a product of a Gaussian density and differences of error functions, so each of the $n_T \cdot n_X \cdot n_Y$ entries of the joint table can be evaluated to arbitrary precision by adaptive quadrature.

The reaction symmetry $\qX(\xi) = \qY(-\xi)$ constrains the PID.
Let $\piT(k) = n_T + 1 - k$ denote the target-bin mirror map and $\piX(k) = n_X + 1 - k$ the source-bin mirror map.
The symmetric charge bins give
$P(Y\!=\!j \mid \xi) = P(X\!=\!\piX(j) \mid -\xi)$;
combined with the symmetric target bins, a change of variable $\xi \to -\xi$ in~\eqref{eq:joint} yields $p_{TY}(t,j) = p_{TX}(\piT(t),\, \piX(j))$.
Because mutual information and specific information are invariant under invertible relabelling, this gives:
\begin{equation}
  I(T;X) = I(T;Y),\qquad
  \Ispec(t;X) = \Ispec(\piT(t);Y) \quad\text{for every }t.
  \label{eq:symm-Ispec}
\end{equation}

Substituting~\eqref{eq:symm-Ispec} into the redundancy~\eqref{eq:Imin} and using $p(t) = p(\piT(t))$ (from symmetric target bins):
\begin{equation}
  \R = \sum_t p(t)\,
  \min\!\big(\Ispec(t;X),\;\Ispec(\piT(t);X)\big).
  \label{eq:R-mirror}
\end{equation}
This equals $I(T;X)$ if and only if $\Ispec(t;X) = \Ispec(\piT(t);X)$ for every~$t$---i.e., the specific information of~$X$ is invariant under the target-bin reflection.
Generically, this invariance does not hold: atom~X is more informative about the side of the reaction coordinate on which it forms a bond ($\xi < 0$) than about the opposite side
($\xi > 0$), so $\Ispec(t;X)$ is larger for bins with $t$ indexing negative~$\xi$ than for their mirror images.
The signed-target PID therefore has equal but \emph{nonzero} unique informations:

\begin{proposition}[Equal unique information for signed target]
\label{prop:U-equal}
Under the reaction symmetry $\qX(\xi) = \qY(-\xi)$, i.i.d.\ observation noise, and shared symmetric bins with $T = \mathrm{bin}(\xi)$:
\begin{equation}
  U_X = U_Y = I(T;X) - \R \;\ge\; 0,
  \label{eq:U-equal}
\end{equation}
where $\R$ is given by~\eqref{eq:R-mirror}.
The four-term decomposition reads
\begin{equation}
  I(T;X,Y) = \R + 2U + S,
  \label{eq:RUST}
\end{equation}
with $U \equiv U_X = U_Y$ and $S = I(T;X,Y) - 2\,I(T;X) + \R$.
\end{proposition}

\begin{proof}
Equation~\eqref{eq:symm-Ispec} and $U_X = I(T;X) - \R$, $U_Y = I(T;Y) - \R$ immediately give $U_X = U_Y$. Non-negativity follows from the minimum in~\eqref{eq:Imin} never exceeding $\Ispec(t;X)$.
\end{proof}

\noindent
The decomposition~\eqref{eq:RUST} is the natural form for a symmetric reaction with a signed target: each charge carries unique information about the side of the reaction it participates in, and both carry it in equal measure.

When unique information is not of interest, or when the target is an \emph{unsigned} observable, a stronger result holds:

\begin{proposition}[Vanishing unique information for even target]
\label{prop:U-zero-even}
Under the reaction symmetry $\qX(\xi) = \qY(-\xi)$, i.i.d.\ observation noise, and shared symmetric charge bins (so that the source-bin mirror map $\piX$ exists): if $T = \mathrm{bin}(|\xi|)$ (an even function of $\xi$), then $\Ispec(t;X) = \Ispec(t;Y)$ for every~$t$, and hence $\R = I(T;X) = I(T;Y)$, $U_X = U_Y = 0$, and
\begin{equation}
  I(T;X,Y) = \R + S,\qquad
  S = I(T;X,Y) - I(T;X) \ge 0.
  \label{eq:RS-even}
\end{equation}
\end{proposition}

\begin{proof}
For an even target $T = \mathrm{bin}(|\xi|)$ with edges
$0 = b_0 < b_1 < \cdots < b_{n_T}$, 
each target bin corresponds to the union
$I_t = [-b_t, -b_{t-1}) \cup [b_{t-1}, b_t)$.
This region is invariant under $\xi \to -\xi$, so
$\piT = \mathrm{id}$.
Then~\eqref{eq:symm-Ispec} becomes $\Ispec(t;X) = \Ispec(t;Y)$ for all~$t$, and the minimum in~\eqref{eq:Imin} is attained at equality.
Non-negativity of $S$ follows from $I(T;X,Y) \ge I(T;X)$.
\end{proof}

\begin{corollary}[Exact charge conservation collapses synergy]
\label{cor:conservation}
If the observed charges satisfy $q_X^{\mathrm{obs}} + q_Y^{\mathrm{obs}} = C$ exactly (zero independent noise) and the charge bins are shared and symmetric, then $Y = \piX(X)$ almost surely, so $I(T;X,Y) = I(T;X)$ and hence $S = 0$, $U_X = U_Y = 0$.
\end{corollary}

\begin{proof}
Under exact conservation, $q_Y^{\mathrm{obs}} = C - q_X^{\mathrm{obs}}$ and symmetric binning yields $Y = \piX(X)$.
Since $\piX$ is an invertible relabelling of the source alphabet, $\Ispec(t;Y) = \Ispec(t;X)$ for every~$t$ (the KL divergence $D_{\mathrm{KL}}(p(Y \mid t) \| p(Y))$ equals $D_{\mathrm{KL}}(p(X \mid t) \| p(X))$ under the bijection $\piX$).
The minimum in~\eqref{eq:Imin} is therefore attained at equality: $\R = \sum_t p(t)\,\Ispec(t;X) = I(T;X)$.
Hence $U_X = U_Y = 0$.
Finally, $I(T;X,Y) = I(T;X)$ since $\piX(X)$ is a function of~$X$, so $S = I(T;X) - 2\,I(T;X) + I(T;X) = 0$.
\end{proof}

\noindent
Corollary~\ref{cor:conservation} shows that any non-trivial structure in the PID---whether unique information or synergy---arises from the independence of the noise terms, representing the information contributed by electronic degrees of freedom beyond the minimal two-state subspace (Remark~\ref{rem:charge-sum}).

\subsection{Redundancy--synergy balance and limiting behaviour}
\label{sec:RS-balance}

The balance between redundancy, unique information, and synergy along the reaction coordinate is governed by how informative each individual charge is about the bond-length asymmetry at a given query point.
Differentiating $\qX(\xi) = q_0 - \delta\,w(\xi)$ gives the charge sensitivity
\begin{equation}
  \frac{d\qX}{d\xi}
  = \frac{\delta\,\lambda^2}
         {2(\xi^2 + \lambda^2)^{3/2}},
  \label{eq:sensitivity}
\end{equation}
which is maximal at the transition state ($|d\qX/d\xi|_{\xi=0} = \delta/(2\lambda)$) and decays as $\delta\lambda^2/(2|\xi|^3)$ in the reactant/product basins.
A dimensionless measure of local informativeness is the signal-to-noise ratio
\begin{equation}
  \mathrm{SNR}(\xi_0) \equiv
  \frac{\delta\,h\,\lambda^2}
       {2\,\sigma_q\,(\xi_0^2 + \lambda^2)^{3/2}},
  \label{eq:SNR}
\end{equation}
which compares the change in mean charge across the kernel window to the observation noise.

Where the SNR is large (near the transition state), the charge profile is approximately linear and each charge individually tracks the bond-length asymmetry; in this regime, the redundancy fraction $\R / I(T;X,Y)$ is large.
Where the SNR is small (deep in the reactant or product basins), each charge individually saturates and carries little information about geometry; any residual information in the joint pair beyond what each individual source provides is attributed to synergy and unique information.
Since the SNR~\eqref{eq:SNR} decreases monotonically with $|\xi_0|$, the redundancy fraction decreases as the query point moves away from the transition state.

To make the limiting behaviour precise, the charge profile can be linearised near the transition state ($\qX \approx q_0 - \delta/2 + \delta\xi/(2\lambda)$), yielding a jointly Gaussian approximation for the \emph{continuous} (pre-binning) variables $(\xi, q_X^{\mathrm{obs}}, q_Y^{\mathrm{obs}})$.
In the continuous Gaussian limit, the target is not discretised into signed bins, and the PID reduces to the form of Proposition~\ref{prop:U-zero-even} with $U = 0$.
The resulting continuous mutual informations, expressed in terms of the dimensionless parameter $\gamma \equiv \delta^2 h^2/(4\lambda^2\sigma_q^2)$, are:
\begin{equation}
  \R_{\mathrm{G}} = \tfrac{1}{2}\log_2(1+\gamma),
  \qquad
  S_{\mathrm{G}} = \tfrac{1}{2}\log_2\!\Bigl(
  \frac{1+2\gamma}{1+\gamma}\Bigr).
  \label{eq:gauss}
\end{equation}
These are continuous-variable surrogates, not the binned PID itself (which is bounded above by $\log_2 n_T$ and generically has $U > 0$ for a signed target).
Nevertheless, the \emph{fractions} they predict are informative.
For $\gamma \gg 1$ (high SNR, near the transition state), the redundancy fraction
$\R_{\mathrm{G}}/(\R_{\mathrm{G}} + S_{\mathrm{G}}) \to 1$:
each charge individually captures nearly all available geometric information.
For $\gamma \ll 1$ (low SNR, in the basins), both $\R_{\mathrm{G}}$ and $S_{\mathrm{G}}$ vanish as $\gamma / (2\ln 2)$ with equal magnitude, so the redundancy fraction approaches $1/2$: the pair provides exactly twice the information of either source alone, split evenly between redundancy and synergy.
The VB parameter $\lambda$ (avoided-crossing width) and the charge-transfer parameter $\delta$ together control the smooth interpolation between these two regimes through the single ratio $\gamma$.

\subsection{Numerical results}
\label{sec:numerical-results}

Figures~\ref{fig:h2-mi} and~\ref{fig:h2-frac} show the local PID evaluated numerically along the reaction coordinate for the analytic VB model of Section~\ref{sec:hamiltonian}, with parameters $\lambda = 1$, $\delta = 0.5$, $q_0 = \delta/2$, $\sigma_q = 0.12$, and $h = 0.5$ ($n_T = n_X = 5$; see Section~\ref{sec:joint-pid} for the binning conventions).
The sources are $X = \mathrm{bin}\!\bigl(q({}^{\mathrm{left}}\mathrm{H})\bigr)$ and $Y = \mathrm{bin}\!\bigl(q({}^{\mathrm{right}}\mathrm{H})\bigr)$, i.e.\ the discretised net charges on the two terminal hydrogen atoms, while the target is $T = \mathrm{bin}(\xi)$ with $\xi = d_X - d_Y$.

\begin{figure}[t]
  \centering
  \includegraphics[width=\columnwidth]{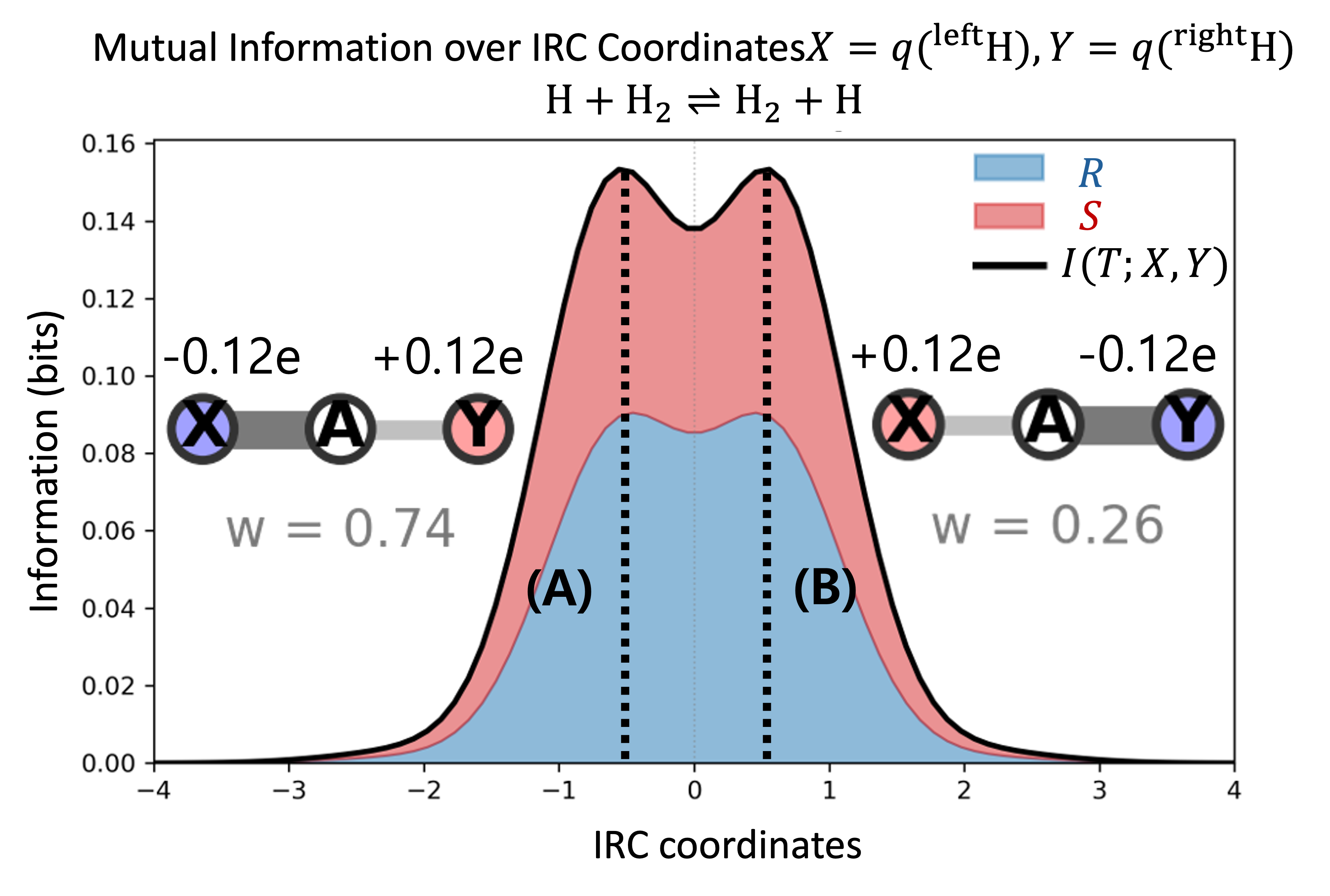}
  \caption{%
    Stacked PID decomposition of the joint mutual information (MI) $I(T;\,X,Y)$ along the IRC for the symmetric atom-transfer reaction $\mathrm{H} + \mathrm{H}_2 \to \mathrm{H}_2 + \mathrm{H}$.
    The blue region is the redundancy~$\R$, the red region is the synergy~$\mathcal{S}$, and the total joint MI $I(T;X,Y)$ (black curve).
    Dashed vertical lines mark the two MI peaks (A) and (B), where the charge magnitudes on the two terminals are equal ($|\bar{q}_X| = |\bar{q}_Y|$).
    Inset atom diagrams show the VB charge distributions at representative points.
    The unique PID contributions are zeroes along the IRC coordinates in this case.}
  \label{fig:h2-mi}
\end{figure}

\begin{figure}[t]
  \centering
  \includegraphics[width=\columnwidth]{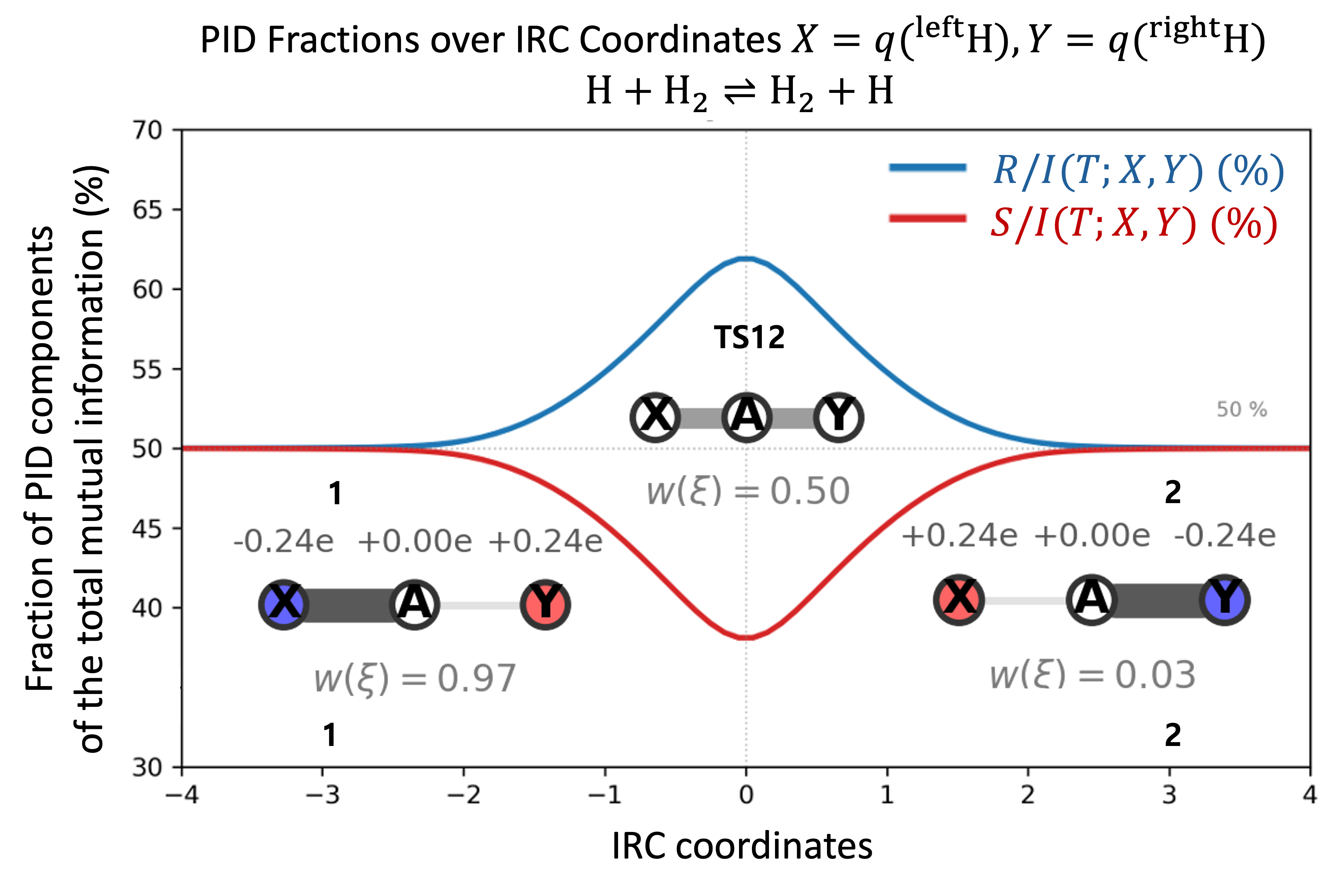}
  \caption{%
    Fractional PID composition $\R/I(T;X,Y)$ and $\mathcal{S}/I(T;X,Y)$ (in percent) along the IRC.
    At the transition state ($\xi = 0$, marked TS12) the redundancy fraction peaks at ${\sim}\,62\%$; deep in the basins (labels~1 and~2) both fractions converge to $50\%$, the low-SNR Gaussian limit~\eqref{eq:gauss}.}
  \label{fig:h2-frac}
\end{figure}

Figure~\ref{fig:h2-mi} displays the joint mutual information $I(T;\,X,Y)$ as a stacked decomposition into redundancy~$\R$ and synergy~$\mathcal{S}$.
The profile exhibits a characteristic double-peaked structure with a local minimum at $\xi = 0$ (the transition state) and two symmetric maxima at $\xi \approx \pm 1$ (marked A and B in the
figure).
At the transition state the charge sensitivity~\eqref{eq:sensitivity} is maximal, yet both terminal charges sit at the midpoint of their sigmoidal profile and yield statistically identical distributions ($\bar{q}_X(0) = \bar{q}_Y(0)$).
Consequently the two sources are maximally redundant: each individually tracks the bond-length asymmetry with high fidelity, but the pair provides little information beyond what either source supplies alone.
The local MI is therefore suppressed relative to the flanking peaks, where the charges have differentiated.
The two peaks at $\xi \approx \pm 1$ correspond to the geometries at which the magnitudes of the terminal charges are equal, $|\bar{q}_X| = |\bar{q}_Y|$, but with opposite signs (one atom gaining electron density as the other loses it).
At these points the charges lie in distinct informational regimes---one is near its bonded limit, the other near its free-radical limit---so the joint source carries substantial synergistic information in addition to a robust redundant component.
This interplay of partial differentiation (enabling synergy) with retained sensitivity (supporting redundancy) maximises $I(T;\,X,Y)$.
In the reactant and product basins ($|\xi| \gg \lambda$) both charges saturate and the MI vanishes, consistent with the $\mathrm{SNR} \to 0$ limit of~\eqref{eq:SNR}: the electronic observables become insensitive to the (now frozen) geometry.

The fractional decomposition in Figure~\ref{fig:h2-frac} sharpens this picture.
At the transition state the redundancy fraction peaks at ${\sim}\,62\%$, confirming that each source individually captures a dominant share of the geometric information (the high-SNR regime, $\gamma \gg 1$, in~\eqref{eq:gauss}).
Because the reaction symmetry enforces $U_X = U_Y$ (Proposition~\ref{prop:U-equal}) and these unique informations are small near the symmetric point, the synergy fraction $\mathcal{S}/I \approx 38\%$ accounts for essentially all the remaining MI at $\xi = 0$.
As $|\xi|$ increases and the SNR falls, the redundancy fraction decreases while the synergy fraction rises; in the deep basins both fractions converge to the $50\%$ asymptote predicted by the low-SNR Gaussian limit~\eqref{eq:gauss} ($\gamma \ll 1$), where the pair provides exactly twice the information of either source alone, split equally between redundancy and synergy.
The monotonic transition between these two limiting behaviours is governed entirely by the dimensionless parameter $\gamma = \delta^2 h^2 / (4\lambda^2 \sigma_q^2)$, confirming that the VB coupling~$\lambda$ and the charge-transfer amplitude~$\delta$ together set the information-theoretic landscape of the reaction.
}{\tbd}
\IfFileExists{results.tex}{
\section{On more complex molecules: \texorpdfstring{S\textsubscript{N}2}{SN2} reactions}
\label{sec:results}

In this section, we consider three S\textsubscript{N}2 reactions:
\begin{itemize}
    \item[(i)] the symmetric identity exchange $\mathrm{F^- + CH_3F \rightleftharpoons CH_3F + F^-}$,
    \item[(ii)] the asymmetric halide substitution $\mathrm{F^- + CH_3Br \rightarrow CH_3F + Br^-}$, and
    \item[(iii)] the hydroxide substitution $\mathrm{OH^- + BrCH_2CH_3 \rightarrow HOCH_2CH_3 + Br^-}$,
\end{itemize}
and analyse how the mutual information and partial information (PI) components, obtained via PID, vary along the intrinsic reaction coordinate (IRC).
Reaction (i) satisfies the permutation symmetry of Proposition~\ref{prop:symmetry}; reaction (ii) breaks this symmetry because the nucleophile and leaving group are chemically distinct; reaction (iii) introduces a larger substrate and a different nucleophile, testing the generality of the framework beyond halide--methyl systems.
All three exhibit the canonical S\textsubscript{N}2 bonding evolution---nucleophilic attack, a pentacoordinate transition state, and leaving-group departure---and are therefore useful benchmarks for PID analysis.

For each reaction, a transition-state (TS) structure was located on the solvated electronic potential energy surface by first computing a minimum-energy path between optimised endpoints using the nudged elastic band (NEB) method~\cite{jonsson1998nudged}, with a climbing-image refinement to converge the saddle point~\cite{henkelman2000climbing,henkelman2000improved}.
The TS was then refined by TS optimisation and verified by harmonic vibrational analysis (a single imaginary frequency along the reaction coordinate).
Starting from the refined TS and its Hessian, the intrinsic reaction coordinate (IRC) was obtained by integrating the mass-weighted steepest-descent equations in both forward and backward directions.

All electronic-structure calculations were performed with ORCA~5.0~\cite{Neese2020,Neese2022} using the $\omega$B97X-D4 density functional~\cite{najibi2020dft,caldeweyher2019generally,chai2008long} and the RIJCOSX approximation~\cite{izsak2011overlap} with the def2/J auxiliary basis~\cite{weigend2006accurate}; solvation was described by the SMD continuum model with water as solvent~\cite{marenich2009universal}.
Def2-SVP~\cite{weigend2005balanced} was used for endpoint optimizations and NEB calculations, while def2-TZVPP~\cite{weigend2005balanced} was used for TS refinement and IRC integration.

Along each IRC, net atomic charges were calculated from the electronic density using the DDEC6 partitioning scheme~\cite{Manz2016,Manz2017} as implemented in Chargemol.
The charges on the nucleophile and leaving-group centres define the electronic sources $X$ and $Y$.
The geometric target is $T=\mathrm{bin}(\xi)$, where the bond-asymmetry coordinate $\xi=d_X-d_Y$ is the difference between the central-carbon distance to the nucleophile ($d_X$) and to the leaving group ($d_Y$) (i.e.\ forming versus breaking C--X distances: C--F versus C--F for the identity exchange, C--F versus C--Br for $\mathrm{CH_3Br+F^-}$, and C--O versus C--Br for $\mathrm{BrCH_2CH_3+OH^-}$).

The PID sweep then follows Algorithm~\ref{alg:pid-irc} with the default parameters specified in Section~\ref{sec:pid-irc}.

\begin{figure}[tb]
  \centering
  \includegraphics[width=\columnwidth]{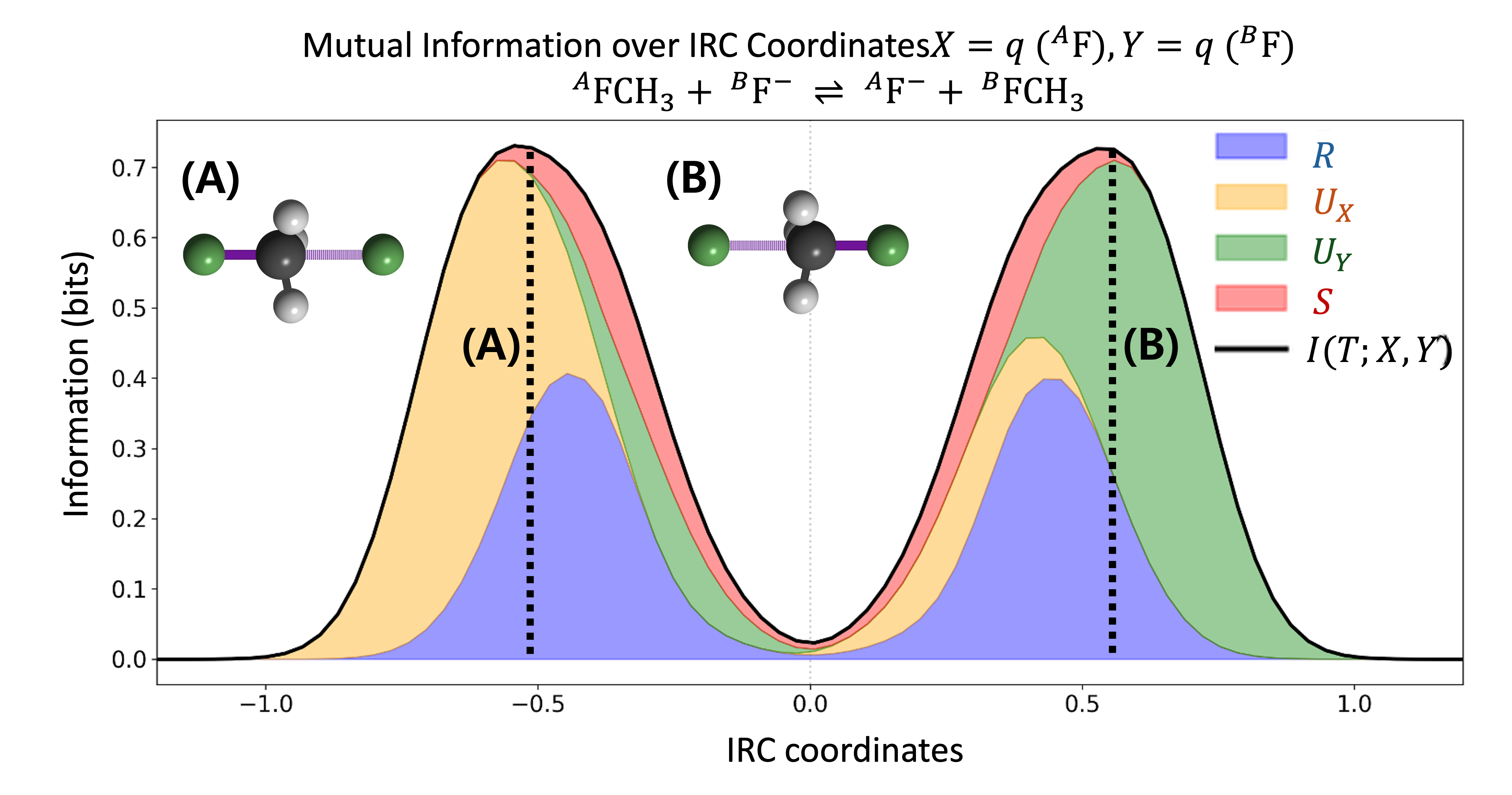}
  \caption{Mutual information and PI components along the IRC for the symmetric identity S\textsubscript{N}2 exchange $\mathrm{F^- + CH_3F \rightleftharpoons CH_3F + F^-}$.
  The black curve shows $I(T;X,Y)$, where the target $T=\mathrm{bin}(d_X-d_Y)$ is the binned bond-asymmetry coordinate and the sources are the DDEC6 charges $X=q({}^{A}\!F)$ and $Y=q({}^{B}\!F)$ on the two labelled fluorine centres.
  Coloured stacked areas show the PI components: redundancy $\R$ (blue), unique informations $U_X$ (orange) and $U_Y$ (green), and synergy $S$ (red), satisfying $I(T;X,Y)=\R+U_X+U_Y+S$ pointwise in $s$.
  Vertical dashed lines mark the symmetric maxima and the insets depict the corresponding mirror-related geometries.}
  \label{fig:fch3f-mi}
\end{figure}

\begin{figure}[tb]
  \centering
  \includegraphics[width=\columnwidth]{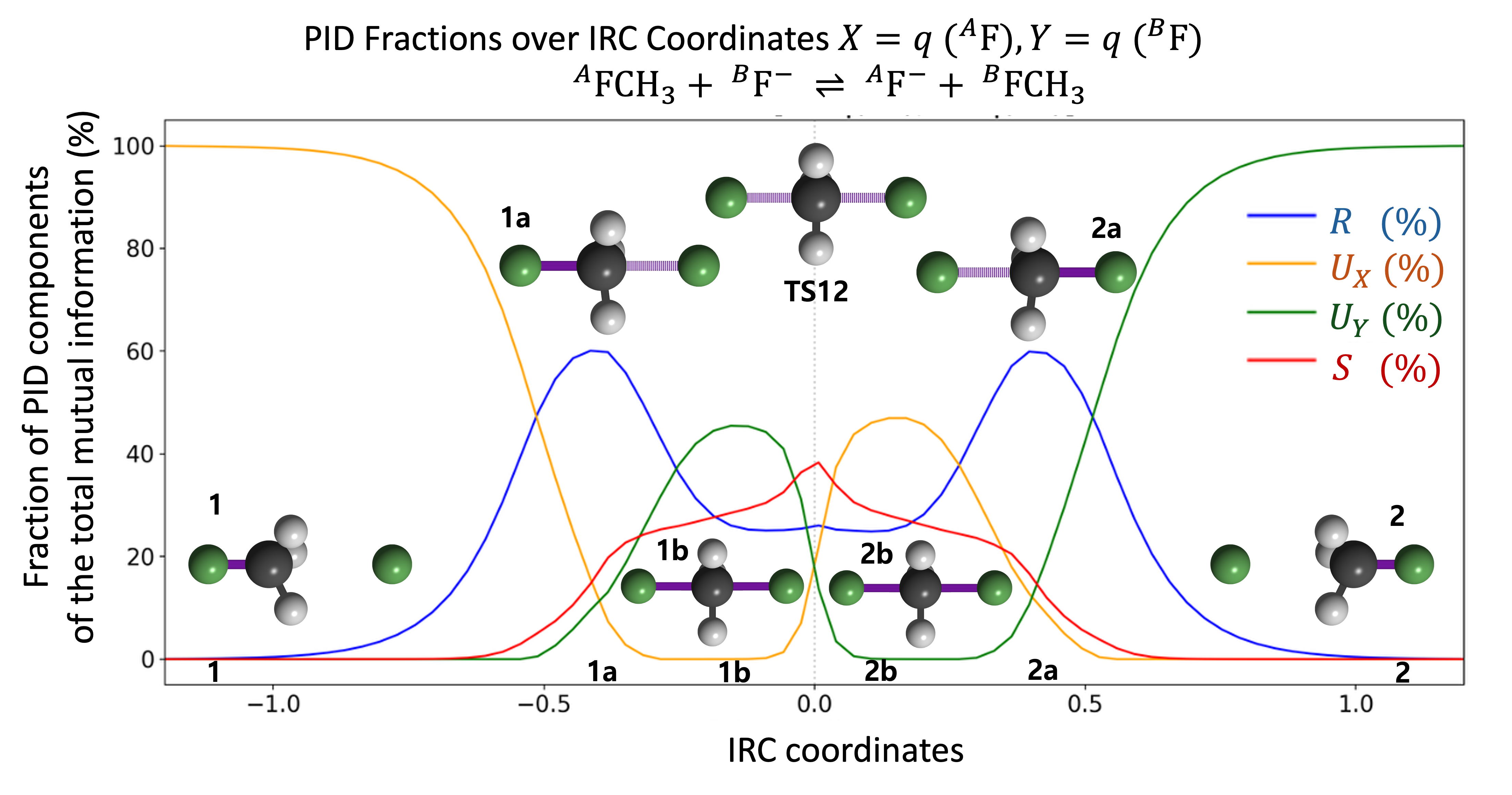}
  \caption{Fractional PI components along the IRC for the symmetric identity S\textsubscript{N}2 exchange $\mathrm{F^- + CH_3F \rightleftharpoons CH_3F + F^-}$ with electronic sources $X=q({}^{A}\!F)$ and $Y=q({}^{B}\!F)$.
  The curves show the PI components normalised by the total joint mutual information, i.e.\ $\R/I(T;X,Y)$ (blue), $U_X/I(T;X,Y)$ (orange), $U_Y/I(T;X,Y)$ (green), and $S/I(T;X,Y)$ (red), expressed as percentages and satisfying $\R+U_X+U_Y+S=I(T;X,Y)$ pointwise in $s$. 
  Representative geometries are shown at the reactant and product extremes (1 and 2), near the mutual-information maxima (1a and 2a), in the intermediate three-centre regime (1b and 2b), and at the transition state (TS12, $s=0$).
  The symmetry of the identity exchange is reflected in the exchange of unique-information fractions under $s\mapsto -s$ and in the near-even profiles of redundancy and synergy.}
  \label{fig:fch3f-frac}
\end{figure}

For the symmetric identity exchange $\mathrm{F^- + CH_3F \rightleftharpoons CH_3F + F^-}$, Fig.~\ref{fig:fch3f-mi} shows the joint mutual information $I(T;X,Y)$ (black curve) between the geometric target $T=\mathrm{bin}(d_X-d_Y)$ and the two electronic sources $X=q({}^{A}\!F)$ and $Y=q({}^{B}\!F)$, together with its PID into redundancy $\R$, unique informations $U_X$, $U_Y$, and synergy $S$ (stacked coloured areas).
Two symmetric maxima of $I(T;X,Y)$ appear at equal magnitude and at opposite IRC coordinates, with mirror-related geometries (insets) and identical total information content, reflecting the permutation symmetry that exchanges the two fluorine labels and reverses the reaction coordinate.
As expected for an identity reaction, the information profile is approximately even in $s$, with component-wise symmetry relations visible in the decomposition: the unique-information contributions interchange under $s\mapsto -s$ (qualitatively, $U_X(s)\approx U_Y(-s)$), while $\R(s)$ and $S(s)$ remain essentially symmetric.

Chemically, the left-hand maximum corresponds to the reactant-side region in which the incoming fluoride begins to interact strongly with the methyl carbon while the original C--F bond is already significantly weakened (pre-TS shoulder), so that small changes in the bond-asymmetry coordinate $\xi=d_X-d_Y$ are strongly reflected in the electronic descriptors.
In this regime the information about $\xi$ is carried predominantly as unique information in $X=q({}^{A}\!F)$, i.e.\ in the charge of the fluorine initially bound to carbon; as the trajectory proceeds toward the central symmetric region, this unique contribution diminishes and is redistributed into shared components (redundancy and synergy), consistent with an increasing degree of coupled response of the two fluorine charges as the three-centre arrangement develops. Past the centre, the same pattern repeats with labels exchanged: the unique information becomes dominated by $U_Y$ on the product-side maximum, corresponding to the fluorine that is now bonded to carbon, while the complementary source becomes less informative. Thus the PID makes the symmetry-enforced ``handoff'' of geometric information between the two fluorine centres explicit: information that is unique to the initially bound atom on one side becomes, after the exchange, unique to the newly bound atom on the other side, with redundancy and synergy peaking in the intermediate regime where the electronic response is intrinsically two-centre.

Figure~\ref{fig:fch3f-frac} reports the \emph{fractional} PI components, i.e.\ each component normalised by the total joint mutual information, $\R/I(T;X,Y)$, $U_X/I(T;X,Y)$, $U_Y/I(T;X,Y)$, and $S/I(T;X,Y)$ (expressed as percentages), for the symmetric identity exchange $\mathrm{F^- + CH_3F \rightleftharpoons CH_3F + F^-}$ with sources $X=q({}^{A}\!F)$ and $Y=q({}^{B}\!F)$.
At the asymptotic ends of the IRC (labels 1 and 2), the information is essentially \emph{purely unique}: on the reactant-side extreme (negative $s$) the geometry target is encoded almost entirely in $U_X$, whereas on the product-side extreme (positive $s$) it is encoded almost entirely in $U_Y$.
This ``one-centre'' regime is expected when one fluorine is clearly bound to carbon and the other is remote, so that the charge on the bound centre is the dominant electronic reporter of the bond-asymmetry coordinate.

Moving from either end toward the TS, the fractional PI components reveal a structured redistribution of information that is constrained by the reaction symmetry.
First, the unique component that dominates at the endpoint decreases and redundancy rises, reaching local maxima near 1a and 2a, coincident with the onset of substantial interaction in the forming/breaking region.
As the system approaches the more three-centre-like geometries (1b and 2b), the initially dominant unique term is strongly suppressed and the complementary unique term becomes appreciable, signalling that the geometric macrostate is no longer attributable to one fluorine charge alone but is shared between the two sources.
At the symmetric TS ($s=0$), the fractional synergy reaches its maximum, indicating that the most informative descriptor of the bond-asymmetry at this point is a genuinely \emph{joint} charge pattern across both fluorine centres rather than either charge marginally.
The right-hand side mirrors the left-hand side under exchange of the fluorine labels, consistent with the expected symmetry relations $U_X(s)\approx U_Y(-s)$ and approximately even $\R(s)$ and $S(s)$.

\begin{figure}[tb]
  \centering
  \includegraphics[width=\columnwidth]{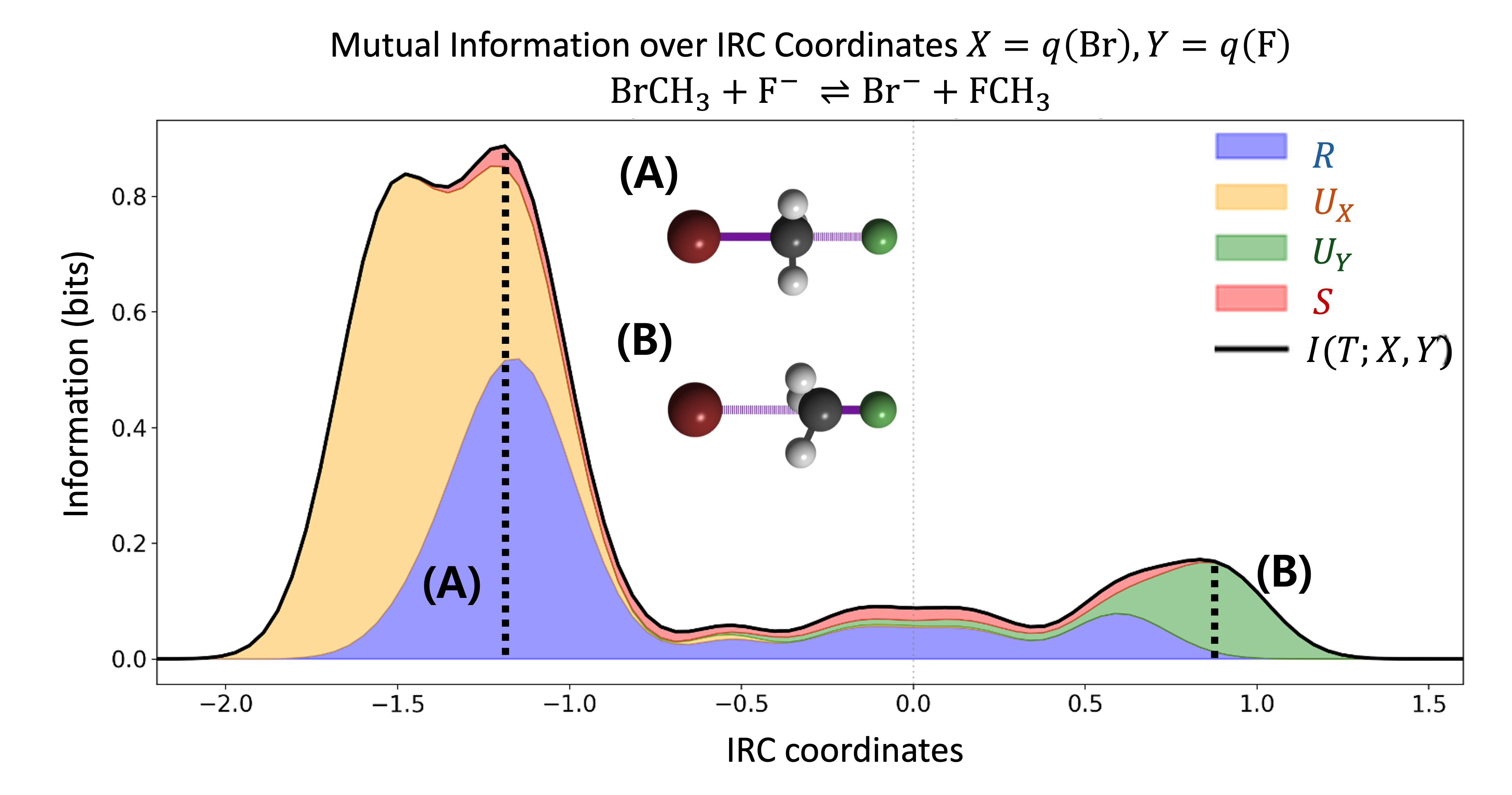}
  \caption{Mutual information and PI components along the IRC for the asymmetric S\textsubscript{N}2 reaction $\mathrm{F^- + CH_3Br \rightarrow CH_3F + Br^-}$.
  The black curve shows $I(T;X,Y)$ with $T=\mathrm{bin}(d_X-d_Y)$ and sources $X=q(\mathrm{Br})$ and $Y=q(\mathrm{F})$ (DDEC6 charges).
  Coloured stacked areas show redundancy $\R$ (blue), unique informations $U_X$ (orange) and $U_Y$ (green), and synergy $S$ (red), satisfying $I(T;X,Y)=\R+U_X+U_Y+S$ pointwise in $s$.
  Dashed vertical lines mark representative maxima on the reactant- and product-side branches; insets depict the corresponding geometries.}
  \label{fig:brch3f-mi}
\end{figure}

\begin{figure}[tb]
  \centering
  \includegraphics[width=\columnwidth]{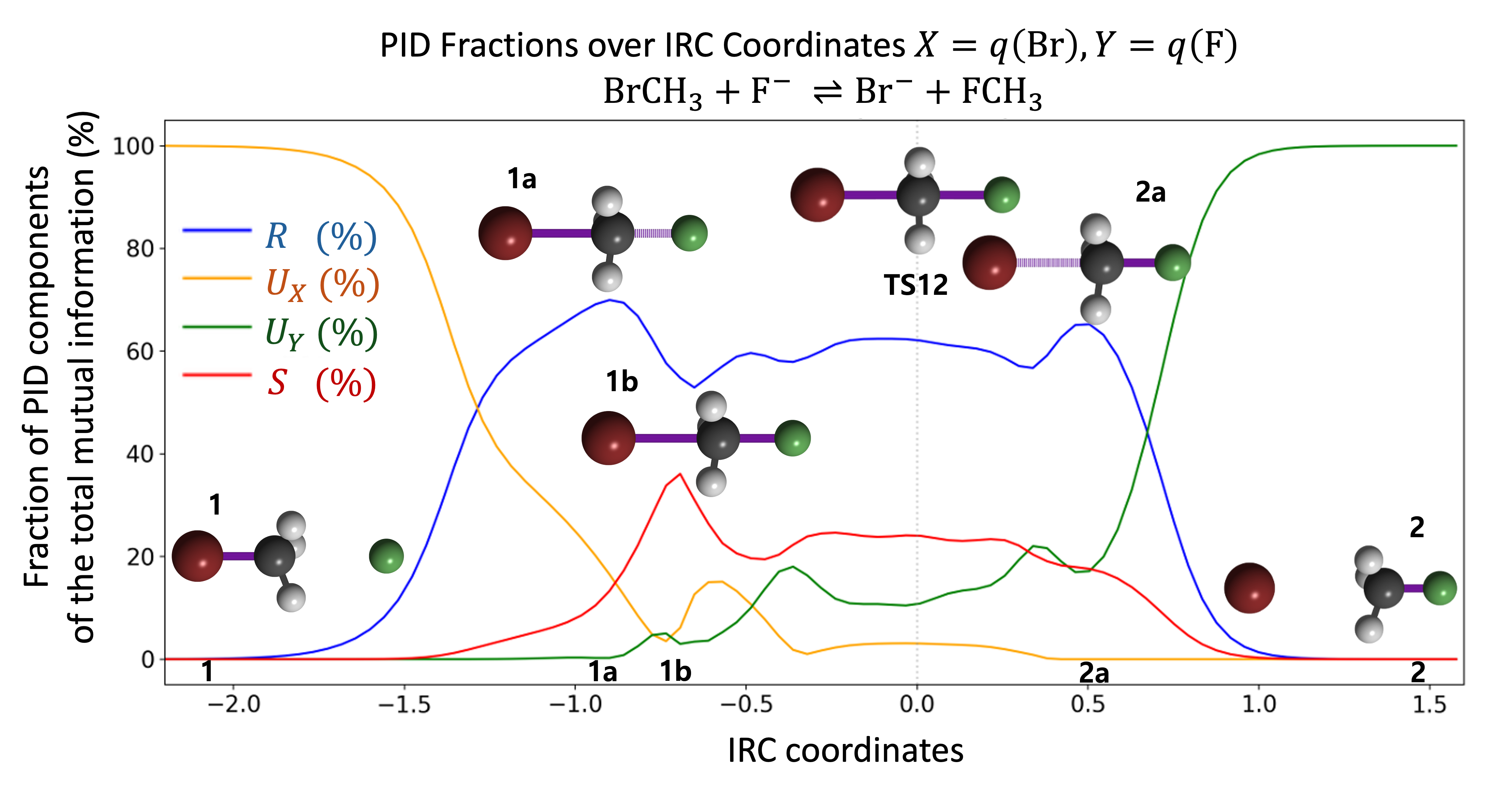}
  \\caption{Fractional PI components along the IRC for the asymmetric S\textsubscript{N}2 reaction $\mathrm{F^- + CH_3Br \rightarrow CH_3F + Br^-}$ with sources $X=q(\mathrm{Br})$ and $Y=q(\mathrm{F})$ (DDEC6 charges).
  Curves show $\R/I(T;X,Y)$ (blue), $U_X/I(T;X,Y)$ (orange), $U_Y/I(T;X,Y)$ (green), and $S/I(T;X,Y)$ (red), expressed as percentages and satisfying $\R+U_X+U_Y+S=I(T;X,Y)$ pointwise in $s$.
  Representative geometries are shown at the reactant and product asymptotes (1 and 2), near the reactant-side redundancy maximum (1a), the intermediate synergistic regime (1b), the transition state (TS12, $s=0$), and the product-side redundancy feature (2a).}
  \label{fig:brch3f-frac}
\end{figure}
For the asymmetric substitution $\mathrm{F^- + CH_3Br \rightarrow CH_3F + Br^-}$, Fig.~\ref{fig:brch3f-mi} displays the joint mutual information $I(T;X,Y)$ (black) and its PID into redundancy $\R$, unique informations $U_X$, $U_Y$, and synergy $S$, with sources $X=q(\mathrm{Br})$ and $Y=q(\mathrm{F})$ and target $T=\mathrm{bin}(d_X-d_Y)$ defined by the forming/breaking C--X distances.
In contrast to the identity exchange (Fig.~\ref{fig:fch3f-mi}), both $I(T;X,Y)$ and its decomposition are strongly \emph{asymmetric} with respect to the IRC coordinate.
On the reactant-side branch (negative $s$), two pronounced maxima appear at $s < -1$, signalling two distinct regimes in which the electronic descriptors most strongly report the evolution of the bond-asymmetry coordinate: an earlier ``onset'' region where the approaching nucleophile begins to polarise the substrate, followed by a sharper region in which the interaction is sufficiently strong that small geometric changes produce a large, highly structured electronic response.

The PI components clarifies how this information is carried.
Near the larger reactant-side maximum, redundancy becomes substantial and the synergy component is non-negligible, indicating that the geometry is best encoded by a \emph{joint} pattern involving both centres rather than by either charge in isolation---consistent with a regime where the forming C--F interaction and the weakening C--Br bond are simultaneously expressed in the electronic redistribution.
By contrast, on the product-side branch (positive $s$) the most prominent peak is much smaller and is dominated by $U_Y$, with only modest redundancy and little synergy.
This indicates that, once the C--F bond is formed and bromide is departing, the geometric target is largely readable from the fluorine-centred descriptor alone, while the leaving-group charge contributes comparatively little additional geometric information except when one approaches the reverse direction.
Overall, the loss of left--right symmetry in the PI components provides a compact information-theoretic signature of mechanistic asymmetry: the nucleophilic-attack side exhibits a pronounced two-centre (redundant/synergistic) encoding, whereas the product-side geometry is predominantly captured by a one-centre unique contribution.

Figure~\ref{fig:brch3f-frac} shows the PID components normalised by the total joint mutual information $I(T;X,Y)$ for the asymmetric substitution
$\mathrm{F^- + CH_3Br \rightarrow CH_3F + Br^-}$, with sources $X=q(\mathrm{Br})$ and $Y=q(\mathrm{F})$.
At the far ends of the IRC (labels 1 and 2), the reacting ion is effectively non-interacting with the substrate: on the reactant side (1) the nucleophile is well separated from $\mathrm{CH_3Br}$, while on the product side (2) $\mathrm{Br^-}$ is well separated from $\mathrm{CH_3F}$.
In these asymptotic regimes the geometric target is encoded almost entirely as \emph{unique} information in a single source (near 100\% $U_X$ or $U_Y$), consistent with a situation where the local bond-asymmetry $\xi=d_X-d_Y$ is essentially determined by one centre at a time (either the leaving group in the reactant basin or the newly formed C--F bond in the product basin).

Approaching from the reactant side toward the TS, the unique fraction $U_X/I$ decreases and a pronounced redundancy maximum develops near 1a, indicating that both charges begin to respond in a correlated way as the nucleophile approaches and the substrate becomes strongly polarised.
A subsequent rise in the synergy fraction (near 1b) signals a regime where the progress variable $\xi$ is best captured by the \emph{joint} charge pattern across the nucleophile and leaving group, rather than by either marginal alone—an information-theoretic signature of the developing three-centre, two-bond arrangement in which C--F formation and C--Br cleavage are coupled.
Across the central region around the TS, the redundancy remains substantial while $U_Y$ gradually increases, reflecting the shift of the most informative charge response from the leaving group to the nucleophile/product side as the C--F bond becomes established and bromide departs.
On the product side, a local redundancy enhancement near 2a marks the regime where the departing bromide still exerts an appreciable influence on the electronic redistribution; for larger positive $s$, this influence vanishes and the fractions collapse back to a single dominant unique contribution in the asymptotic ion–molecule limit.

\begin{figure}[tb]
  \centering
  \includegraphics[width=\columnwidth]{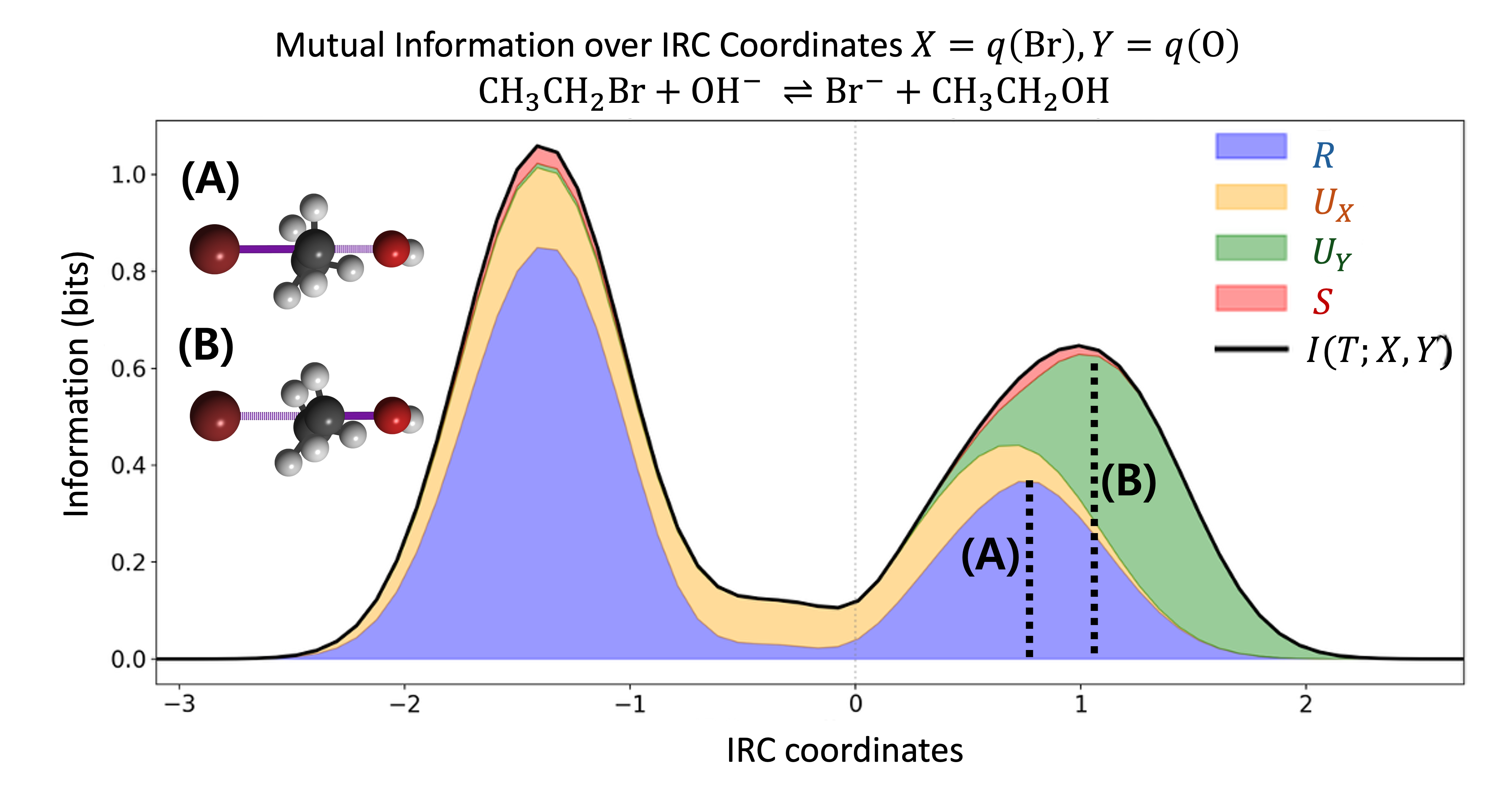}
  \caption{Mutual information and PID along the IRC for the S\textsubscript{N}2 reaction $\mathrm{CH_3CH_2Br + OH^- \rightarrow CH_3CH_2OH + Br^-}$.
  The black curve shows the joint mutual information $I(T;X,Y)$ between the geometric target $T=\mathrm{bin}(d_{\mathrm{C}\!-\!\mathrm{O}}-d_{\mathrm{C}\!-\!\mathrm{Br}})$ and the electronic sources $X=q(\mathrm{Br})$ and $Y=q(\mathrm{O})$ (DDEC6 charges).
  Coloured stacked areas show the PID components: redundancy $\R$ (blue), unique informations $U_X$ (orange) and $U_Y$ (green), and synergy $S$ (red), satisfying $I(T;X,Y)=\R+U_X+U_Y+S$ pointwise in the IRC coordinate.
  Dashed vertical lines mark representative product-side configurations (A) and (B); the corresponding inset geometries illustrate the relative extent of C--Br cleavage and C--O formation in these regimes.}
  \label{fig:bromoethane-mi}
\end{figure}

\begin{figure}[tb]
  \centering
  \includegraphics[width=\columnwidth]{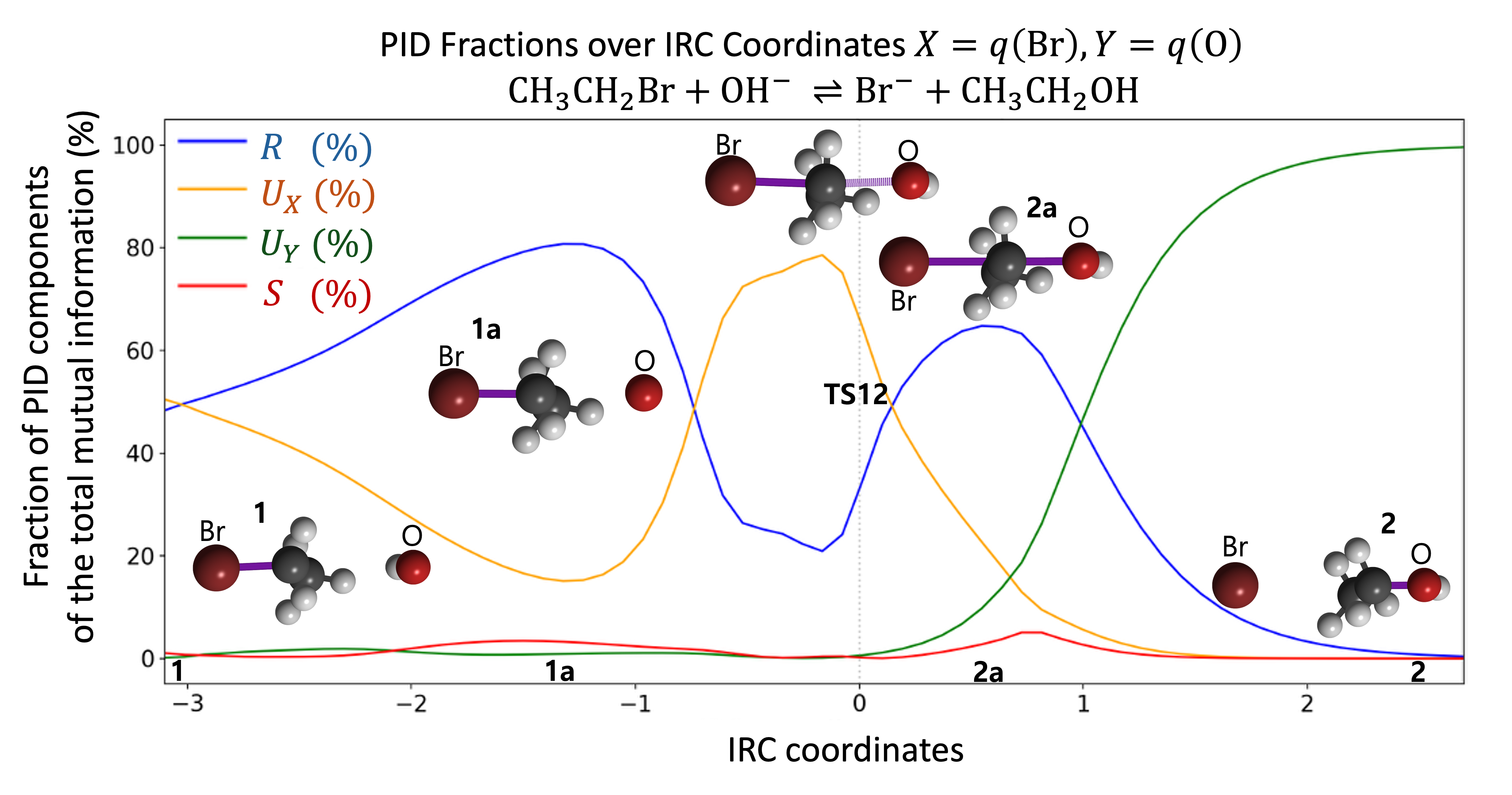}
  \caption{Fractional PI components along the IRC for $\mathrm{CH_3CH_2Br + OH^- \rightarrow CH_3CH_2OH + Br^-}$ with sources $X=q(\mathrm{Br})$ and $Y=q(\mathrm{O})$ (DDEC6 charges).
  Curves show $\R/I(T;X,Y)$ (blue), $U_X/I(T;X,Y)$ (orange), $U_Y/I(T;X,Y)$ (green), and $S/I(T;X,Y)$ (red), expressed as percentages and satisfying $\R+U_X+U_Y+S=I(T;X,Y)$ pointwise in the IRC coordinate.
  Representative geometries are shown at the reactant and product asymptotes (1 and 2), near reactant- and product-side redundancy maxima (1a and 2a), and at the transition state (TS12, $s=0$).}

  \label{fig:bromoethane_frac}
\end{figure}

For the larger substrate reaction $\mathrm{CH_3CH_2Br + OH^- \rightarrow CH_3CH_2OH + Br^-}$, Fig.~\ref{fig:bromoethane-mi} exhibits a strongly asymmetric information profile with two maxima in $I(T;X,Y)$.
The dominant peak on the reactant-side branch (negative IRC) reflects the fact that nucleophilic attack by $\mathrm{OH^-}$ induces a particularly strong and extended electronic reorganisation while the C--Br bond is still largely intact: as the oxygen lone pair approaches the electrophilic carbon, the incipient C--O interaction strongly polarises the substrate and simultaneously weakens and polarises the C--Br bond.
Because the target $T=\mathrm{bin}(d_{\mathrm{C}\!-\!\mathrm{O}}-d_{\mathrm{C}\!-\!\mathrm{Br}})$ couples these two distances, the corresponding charge responses on both centres,
$X=q(\mathrm{Br})$ and $Y=q(\mathrm{O})$, co-vary appreciably over this region, producing the large joint mutual information and its predominantly redundant character.
In other words, the larger left peak is not because bromine ``attracts'' carbon more strongly; rather it indicates that the approaching hydroxide perturbs the electronic structure over a wider range of geometries (strong polarisation and bond reorganisation), so both source charges become informative about the bond-asymmetry well before the transition state.

On the product-side branch (positive IRC), the smaller peak reflects a late-stage regime where the departing bromide still influences the electronic redistribution, but progressively decouples as separation increases.
At the point marked (A), the redundancy fraction is locally maximised, consistent with a configuration in which both the weakening C--Br interaction and the developing C--O bond contribute comparably to the geometric macrostate, so that both $q(\mathrm{Br})$ and $q(\mathrm{O})$ respond in a correlated way.
By the time the local maximum of $I(T;X,Y)$ is reached (B), the balance shifts: the C--O bond is more established and the oxygen-centred descriptor carries a larger unique contribution, while $q(\mathrm{Br})$ approaches its asymptotic anionic value and becomes less sensitive to further geometric change.
Beyond this region, as $\mathrm{Br^-}$ and $\mathrm{CH_3CH_2OH}$ become effectively non-interacting, all PID components (and hence $I(T;X,Y)$) decay toward zero, consistent with the loss of electronic sensitivity to the bond-asymmetry coordinate in the separated-products limit.

Figure~\ref{fig:bromoethane_frac} reports the fractional PI components for $\mathrm{CH_3CH_2Br + OH^- \rightarrow CH_3CH_2OH + Br^-}$, i.e.\ each component normalised by the local joint mutual information $I(T;X,Y)$, with $X=q(\mathrm{Br})$ and $Y=q(\mathrm{O})$.
Compared with the smaller halide--methyl systems, the approach to the asymptotic reactant region (label 1) is less sharply ``one-centre'': the redundancy fraction decays more gradually and $U_X/I$ does not saturate near 100\% (the IRC coordinate needs to hit -5 to be completely unique).
This reflects the fact that, for a larger and more flexible substrate, the bond-asymmetry coordinate $\xi=d_{\mathrm{C}\!-\!\mathrm{O}}-d_{\mathrm{C}\!-\!\mathrm{Br}}$ can still correlate weakly with broader substrate polarisation and conformational degrees of freedom even when the nucleophile is relatively distant, so the leaving-group charge does not act as an isolated reporter of $T$.

As the system proceeds from the reactant side toward the transition region, $\R/I$ increases and reaches a maximum near 1a, consistent with the onset of a strongly coupled electronic response in which both $q(\mathrm{Br})$ and $q(\mathrm{O})$ track the evolving approach geometry.
Chemically, this regime corresponds to the nucleophile adopting an orientation and distance that optimally engages the electrophilic carbon (backside approach) while the C--Br bond begins to polarise and weaken; the two electronic descriptors therefore carry overlapping information about $\xi$.
Around the TS ($s\approx 0$), the fractions indicate a redistribution from redundancy toward a larger unique contribution in $X$ on the early product side, followed by a rapid rise in $U_Y/I$ as the C--O bond becomes established and the oxygen-centred descriptor becomes the dominant reporter of the bond-asymmetry.
On the positive branch, the local redundancy maximum near 2a marks the late-stage three-centre regime in which C--O formation and C--Br cleavage remain coupled; beyond this point, as bromide separates and approaches its asymptotic anionic character, its contribution to the information rapidly diminishes and the decomposition collapses toward predominantly unique information in $Y$ (label 2).}{\tbd}
\IfFileExists{discussions.tex}{
\section{Discussion}
\label{sec:discussion}

This work introduced a reaction-coordinate-resolved partial information decomposition (PID) that quantifies, at each point along an intrinsic reaction coordinate (IRC), how selected electronic descriptors encode a coarse-grained geometric progress variable.
Across the S\textsubscript{N}2 benchmarks, the decomposition reproduces chemically sensible structure: symmetry-enforced exchange of the unique-information profiles in the identity reaction, systematic asymmetry when nucleophile and leaving group are distinct, and a broader redistribution of information in the larger bromoethane system where additional internal degrees of freedom influence the electronic response.
Viewed as a diagnostic, the PID provides a compact language for describing \emph{how} the electronic reorganisation along a reaction path is expressed in the chosen observables: as information that is accessible from either source alone (redundancy), attributable predominantly to one source (unique information), or encoded primarily in their joint pattern (synergy).
At the same time, the PID is a statement about the local empirical joint distribution $p(t,x,y\mid s_0)$ at each $s_0$ and should be interpreted as mechanistically suggestive rather than causal.

\paragraph{Limitations.}
Several limitations delimit the current scope and clarify where care is required.
First, the present formulation emphasises a fixed pair of sources $(X,Y)$ tied to two atomic centres.
In multi-centre scenarios (e.g.\ agostic interactions or delocalised rearrangements), the relevant electronic response may not be localisable to a single obvious pair of atoms; an adequate analysis may require repeating the PID over multiple centre choices or adopting fragment- or density-based sources.
Second, the pipeline conditions on a preselected reactant--product pair and analyses a single minimum-energy path.
This is appropriate for illustrating the method, but it does not address \emph{competition} between mechanisms.
In particular, for $\mathrm{CH_3CH_2Br + OH^-}$ the experimentally relevant pathway can include E2 elimination; our present analysis characterises the substitution IRC given fixed endpoints but does not quantify the relative accessibility of alternative channels or branching between them.
Third, the method relies on discretising continuous observables (``pixelisation'') to form discrete variables.
While discretisation is essential for the discrete PID, it necessarily coarsens the physical picture and can introduce sensitivity to bin resolution, tail behaviour, and to the chosen partitioning scheme for the electronic observable (e.g.\ different population analyses can shift quantitative charge values).
This issue becomes more acute for larger systems, where a single atomic charge may be an insufficient summary of a distributed electronic response.
Moreover, discretisation currently limits direct comparison between information measures and continuous mechanical quantities, such as reaction forces and energy gradients, without additional modelling to define smooth information profiles and their derivatives.
Finally, our demonstrations focus on reactions whose progress is naturally parameterised by a single bond-asymmetry coordinate about one carbon centre.
Highly concerted transformations involving strongly coupled multi-bond rearrangements (e.g.\ collective cyclisations or multi-centre bond reorganisations) are not addressed here; such cases likely require multi-dimensional targets and/or more-than-two sources to capture higher-order cooperative structure.

\paragraph{Future directions.}
These limitations suggest several concrete extensions.
A first priority is to connect information-theoretic structure to energetic diagnostics along the reaction coordinate.
On a sufficiently smooth representation of the local distributions (or on continuous analogues of the information measures), one can define derivatives of $I(T;X,Y)$ and of the PID components with respect to $s$, and compare their extrema to features of the potential-energy profile and the reaction force.
Such a comparison would sharpen the interpretation of ``onset'' and ``handoff'' regions by linking changes in informational encoding to changes in the underlying energetics.
A second direction is to expand the framework to address multi-centre and environment-dependent chemistry.
Practically, this motivates (i) promoting $X,Y$ beyond atomic charges to fragment charges, bond indices, or density-derived descriptors; (ii) performing systematic robustness checks across chemically sensible partitions and electronic population schemes; and (iii) moving from a single IRC trace to an ensemble conditioned on $s$ (or on a collective variable), thereby quantifying not only mean PID profiles but also their variability under thermal and environmental fluctuations.
A third direction is to extend beyond the two-source setting, enabling treatment of concerted reactions where progress is encoded in higher-order electronic patterns that cannot be reduced to a single pair of centres.
Finally, it is natural to recast the present construction in a channel/measurement language familiar from quantum information: the map from nuclear configuration to electronic descriptors defines an effective information channel, and the PID compares how different ``measurements'' (descriptor choices and partitions) encode the chosen geometric macrostate.
Such a reformulation could provide a more physics-oriented viewpoint, clarify the role of coarse-graining as a measurement model, and suggest principled ways to compare descriptor families and partitions on equal footing.}{\tbd}
\IfFileExists{conclusion.tex}{
\section{Conclusion}
\label{sec:conclusion}

In this work, a reaction-coordinate-resolved partial information decomposition (PID) framework is introduced to quantify, along an intrinsic reaction coordinate, how selected electronic descriptors encode a coarse-grained geometric progress variable.
By decomposing the joint mutual information into redundant, unique, and synergistic components, the method distinguishes regimes where reaction progress is readable from either centre alone, localised primarily on one centre, or expressed predominantly as a joint electronic pattern.

Applied to three prototypical S\textsubscript{N}2 reactions, the PID recovers chemically intuitive structure: symmetry-enforced exchange of unique-information profiles in the identity reaction, systematic asymmetry when nucleophile and leaving group differ, and broader information redistribution in the larger bromoethane system where additional molecular degrees of freedom influence the electronic response.
These results demonstrate that PID provides an interpretable and symmetry-sensitive descriptor of electronic reorganisation along a reaction path, complementary to standard energetic and structural analyses.

More broadly, the framework offers a principled route to compare electronic observables and partitions by their information content with respect to a chosen progress variable.
Future developments will focus on robustness across descriptor families and partitions, ensemble-based treatments and competing pathways, and differential comparisons to reaction forces and energy gradients, with the longer-term aim of integrating information-theoretic and physics-based perspectives on chemical reactivity.}{\tbd}

\bibliographystyle{unsrtnat} 
\bibliography{refs.bib}    

\end{document}